\documentclass[a4paper,11pt]{article}

\usepackage[T1]{fontenc}
\usepackage[utf8]{inputenc}
\usepackage[english]{babel}
\usepackage[normalem]{ulem} 

\usepackage{geometry}
\usepackage{amsmath, amssymb, amsfonts, bm}
\usepackage{mathalfa}
\usepackage{mathtools}

\usepackage{graphicx}
\usepackage{subcaption}
\usepackage{booktabs}
\usepackage{array}
\usepackage{tabularx}
\usepackage{longtable}
\usepackage{threeparttable}
\usepackage{float}
\usepackage{standalone}
\usepackage{tikz}

\usepackage[table]{xcolor}
\definecolor{myblue}{RGB}{0, 114, 178}  
\definecolor{myred}{RGB}{213, 94, 0}  
\definecolor{mygray}{RGB}{80, 80, 80}

\usepackage{algorithm}
\usepackage{algpseudocode}

\usepackage{hyperref}
\usepackage[round,longnamesfirst]{natbib}
\usepackage{xspace}
\usepackage{authblk}
\usepackage{comment}
\usepackage{pdflscape}
\usepackage{rotating}
\usepackage{appendix}
\usepackage{subfiles}

\usepackage{amsthm}
\newtheorem{assumption}{Assumption}
\newtheorem{theorem}{Theorem}

\newtheorem{proposition}{Proposition}

\begin{document}

\title{\Large{\textbf{Group-Level Treatment Effect Heterogeneity in Difference-in-Differences: A Balanced Approach}}}

\date{\today}
\author{Nora Bearth\footnote{ Swiss Institute of Empirical Economic Research, University of St. Gallen.}, \ Nadja van 't Hoff\footnote{Amsterdam School of Economics, University of Amsterdam.}, and \ Torben S. D. Johansen\footnote{Department of Economics, University of Southern Denmark}}
\maketitle

\begin{abstract} \noindent
Understanding how treatment effects vary across groups is central to policy evaluation. In Difference-in-Differences designs, heterogeneity is often studied using subgroup or triple-difference analyses, which can suffer from conservative inference, reliance on parametric interaction structures, and sensitivity to differences in covariate distributions across groups. We propose the \textit{Balanced Group Average Treatment Effect on the Treated} (BGATT), a new estimand that isolates heterogeneity in treatment responses from differences in covariate composition and is identified under standard conditional parallel-trends assumptions. BGATT provides a transparent target for comparing group-specific treatment effects. We derive an influence-function representation and develop estimators that are $\sqrt{n}$-consistent and asymptotically normal under flexible machine-learning estimation of high-dimensional nuisance components, enabling valid inference on both group-specific effects and differences across groups. Simulation evidence shows favorable finite-sample performance.
\bigskip 
 
\end{abstract}

\vfill

{\small \renewcommand{\thefootnote}{\arabic{footnote}} %
\setcounter{footnote}{0} \pagebreak \setcounter{footnote}{0} \pagebreak %
\setcounter{page}{1} }

% --------------------------------------------------
% Introduction
% --------------------------------------------------
\section{Introduction}
\label{sec:introduction}

Difference-in-differences (DiD) provides a widely used framework for causal inference in policy evaluation. Under a parallel-trends assumption, DiD identifies average treatment effects by comparing outcome changes for treated and comparison units. In many applications, however, average effects alone are insufficient: researchers are often interested in how treatment effects vary with observable characteristics. Such heterogeneity is central for welfare analysis and policy design, for example when treatment effects may differ between men and women, or along dimensions such as education and income.

Existing approaches to studying heterogeneous effects in DiD designs face important limitations. In applied work, researchers typically rely either on subgroup analyses or on interaction-based specifications such as triple differences. Subgroup analyses are intuitive, but conditioning on subgroup membership reduces effective sample sizes, leading to imprecise subgroup-specific estimates and low-powered tests of differences across groups. Moreover, when covariate distributions differ across subgroups, focusing on conditional ATTs along one or two covariate dimensions can conflate the targeted heterogeneity with other interactive effects, underscoring the need for more systematic methodological guidance \citep{baker2025difference}. As a result, conventional subgroup comparisons do not cleanly separate heterogeneity in treatment responses from heterogeneity in the characteristics of treated units.

Triple-difference specifications offer a regression-based alternative, but their coefficients need not coincide with the heterogeneity parameter of interest without additional parametric restrictions \citep{strezhnev2023decomposing}. In flexible DiD settings, interaction-based estimators may implicitly aggregate underlying group-time treatment effects using non-convex, and sometimes negative, weights, making the resulting estimand difficult to interpret. These concerns are compounded when (i) heterogeneity is studied along multiple characteristics or (ii) parallel trends holds only conditional on covariates. In such cases, regression-based implementations require interactions between subgroup indicators and controls, and often higher-order interactions among controls, to avoid restrictive functional-form assumptions. The resulting specifications can easily become high-dimensional, increasing sample-size requirements and limiting practical applicability.

These limitations motivate methods that (i) deliver interpretable heterogeneity parameters tied to explicit causal or descriptive contrasts, (ii) separate variation in treatment responses from differences in covariate distributions across groups, (iii) allow for flexible, high-dimensional covariate adjustment under conditional parallel trends, and (iv) avoid the weighting pathologies associated with traditional two-way fixed-effects and triple-difference specifications. This paper develops such a framework for DiD designs.

We develop an approach to analyzing treatment effect heterogeneity in DiD that explicitly separates treatment-response heterogeneity from differences in covariate composition across groups. We introduce the \textit{Balanced Group Average Treatment Effect on the Treated} (BGATT), a group-specific estimand that recovers treatment effects after balancing the distribution of selected covariates across groups. By construction, BGATT evaluates the average treatment effect for each group under a common target distribution of a subset of covariates $W$, which we take as the marginal distribution of $W$ among the treated. Differences in BGATTs across groups, which we term the \textit{difference in BGATTs} (DiBGATT), then quantify how treatment effects differ across groups once their covariate composition has been equalized. 

Our first contribution is conceptual. We formalize BGATT and DiBGATT in a two-period DiD setting under standard conditional parallel-trends, no-anticipation, overlap, and exogeneity assumptions, and show that BGATT is point-identified from observed data. We then decompose differences in group-specific ATTs into a component that reflects genuine differences in treatment responses and a component driven by differences in covariate composition. This decomposition is analogous to the classic Kitagawa–Oaxaca–Blinder decomposition, but adapted to a potential-outcomes DiD framework and made symmetric by using the marginal treated distribution of covariates as a common reference. This clarifies what part of observed group gaps is due to treatment-effect heterogeneity versus compositional differences.
We discuss the additional assumptions required for DiBGATT to admit a causal interpretation as an effect of the moderator $Z$ itself; in the baseline analysis, we view DiBGATT primarily as a balanced group contrast in treatment effects.

Our second contribution is methodological. We derive the semiparametric efficient influence function for BGATT and use it to construct estimators that remain valid when nuisance components—such as outcome regressions, treatment probabilities, and group-assignment probabilities—are estimated flexibly using modern machine-learning methods. The influence-function representation extends doubly robust DiD scores to the balanced-group setting by combining (i) residualized DiD terms that align treated and comparison units within groups, and (ii) additional reweighting terms that enforce covariate balance with respect to the target distribution of $W$. This structure ensures Neyman orthogonality, so that small regularization errors in the nuisance estimates have only second-order effects on the target parameter. Under high-level product-rate conditions on the nuisance estimators, we show that the resulting BGATT and DiBGATT estimators are $\sqrt{n}$-consistent, asymptotically normal, and attain the semiparametric efficiency bound. The influence-function structure further implies a novel ``cross-double robust’’ property: consistency obtains if, for each of two key pairs of nuisance components, at least one element in the pair is consistently estimated.

Our third contribution is to extend the BGATT framework to multiple time periods with staggered treatment adoption. Following \citet{callaway2021difference}, we define group-time BGATT parameters that measure the effect of first receiving treatment in period $g$ on outcomes in period $t$ for each group $Z=z$, again evaluated under a common balancing distribution of $W$. We show that these group-time BGATTs are identified under generalized conditional parallel trends, no anticipation, overlap, and exogeneity, and we derive their influence functions. The group-time effects can be aggregated into familiar summaries—such as event-study profiles, group-specific averages, calendar-time effects, and overall averages—using standard weighting schemes. We develop asymptotic theory for the joint distribution of group-time and aggregated BGATTs and use a multiplier bootstrap to obtain valid simultaneous confidence bands.

The remainder of the paper is organized as follows. Section \ref{sec:literature} provides overview of related literature. Section \ref{sec:setup} introduces the setup, notation, and formal definition of BGATT. Section \ref{sec:identification_decomposition} establishes identification, presents the decomposition of BGATT, and clarifies its interpretation. Section \ref{sec:influence_function_estimation} derives the influence-function representation, develops the estimator and presents the asymptotic theory and inference results. Section \ref{sec:extension_multiple_periods} extends the results to multiple time periods and staggered treatment adoption. Section \ref{sec:simulations} reports simulation evidence on finite-sample performance. Section \ref{sec:application} applies the proposed methods in an application. Section \ref{sec:conclusion} concludes.

\bigskip

% --------------------------------------------------
% Literature Review
% --------------------------------------------------
\section{Literature Review}
\label{sec:literature}

This paper relates to three strands of econometric research: semiparametric and doubly robust methods for Difference-in-Differences, machine-learning approaches to treatment-effect heterogeneity, and recent work on triple-difference designs.

A large literature studies DiD designs with covariates under conditional parallel trends. \citet{abadie2005semiparametric} provides an early semiparametric formulation of DiD with covariate adjustment. More recent contributions develop efficient and doubly robust estimators for average treatment effects in DiD settings, including estimators that accommodate flexible or high-dimensional nuisance-function estimation \citep{chang2020double,zimmert2018efficient,sant2020doubly,callaway2021difference,nie2021nonparametric}. Relatedly, \citet{huber2026difference} use double machine-learning tools to study mediation effects in DiD settings, while we study moderation effects. These papers substantially broaden the scope of DiD methods by allowing researchers to adjust for covariates and, in staggered-adoption designs, to recover interpretable group-time average treatment effects. Our focus is complementary: rather than treating covariates only as controls for identification or precision, we study how covariate distributions enter the interpretation of group-level treatment-effect comparisons.

For estimation, we derive the efficient influence function for the proposed balanced group treatment-effect parameters and use it to construct doubly robust, Neyman-orthogonal estimating equations that accommodate flexible nuisance estimation (e.g., see \cite{kennedy2024semiparametric}). Unlike much of the recent doubly robust DiD literature, which typically starts from regression, inverse-probability-weighted, or orthogonal scores for standard DiD parameters, our starting point is the semiparametric representation of a new balanced estimand. 

Our paper is also related to work on causal machine learning and conditional treatment-effect estimation. In cross-sectional settings, causal forests, generic machine-learning approaches to heterogeneous treatment effects, and debiased machine-learning methods have been developed to estimate conditional or individualized causal effects \citep{wager2018estimation,athey2019generalized,chernozhukov2018generic}. Related ideas have recently been adapted to DiD settings, where the goal is to estimate conditional treatment effects under conditional parallel trends \citep{gavrilova2023dynamic,hatamyar2023machine,imai2026doubly}. These methods are useful for describing heterogeneity as a function of observed covariates, identifying units or covariate profiles for which effects are large, and constructing conditional group-time treatment-effect functions. 

Group comparisons raise a distinct issue. A raw difference in group-specific ATTs reflects both differences in treatment responses and differences in the covariate distributions of the groups being compared. When treatment effects vary with covariates, a larger ATT may arise because group status moderates the treatment response or because the group contains more units with covariate profiles associated with larger effects. Conditional treatment-effect methods describe how effects vary with covariates, but they do not by themselves define a group-level comparison under a common covariate distribution. This distinction is closely related to the balanced causal machine-learning framework of \citet{bearth2024causal}, who study balanced group average treatment effects in cross-sectional settings.\footnote{See also the large literature on group and conditional treatment effects under unconfoundedness, which studies heterogeneity using balancing and orthogonalization tools in cross-sectional settings: \cite{abrevaya2015estimating, athey2019generalized, lechner2018modified, semenova2021debiased, zimmert2019nonparametric, fan2022estimation, jacob2019group}} We extend this balancing logic to DiD designs by comparing group-specific treatment effects after equalizing the distribution of observed covariates.

The paper also contributes to the literature on triple-difference designs. Triple differences have a long history in applied work, with \citet{gruber1994incidence} providing an early and influential example. Recent methodological work has clarified the identifying assumptions, estimands, and weighting properties of DDD estimators \citep{olden2022triple,strezhnev2023decomposing,ortiz2025better,caron2025triple}. \citet{strezhnev2023decomposing} shows that conventional triple-difference regressions can aggregate underlying comparisons with difficult-to-interpret weights under staggered adoption and heterogeneous effects.\footnote{This is similar to issues found with TWFE models, see for example \cite{goodman2021difference, de2020two}.} \citet{ortiz2025better} develop covariate-adjusted DDD estimators for designs in which treatment effectively requires satisfying two criteria, so that the third dimension is structural to the treatment assignment mechanism. In our setting, by contrast, the group indicator is an analytical dimension chosen to study heterogeneity across substantively meaningful groups.

The closest paper to ours is \citet{caron2025triple}, who studies triple-difference designs with heterogeneous treatment effects and distinguishes between the difference in average treatment effects on the treated across subgroups (DATT) and a causal difference in average treatment effects on the treated (CDATT). Her framework clarifies when a difference in subgroup-specific ATTs can be interpreted causally as the effect of subgroup status itself. CDATT requires assumptions about counterfactual subgroup membership and potential outcomes under hypothetical reassignment of subgroup status. This is natural in some applications, but can be conceptually demanding or inappropriate when subgroup status is immutable or not meaningfully manipulable, such as gender, race, or other demographic characteristics.

Our DiBGATT answers a different question. Whereas \citeauthor{caron2025triple}'s (\citeyear{caron2025triple}) CDATT asks whether differences in treatment effects can be causally attributed to subgroup status, DiBGATT asks whether a treatment-effect gap remains after placing groups on a common covariate distribution. Rather than requiring a causal interpretation of subgroup membership itself, it evaluates group-specific treatment effects under the same covariate distribution and distinguishes the portion of an observed ATT gap that is driven by differences in covariate composition from the portion that remains after balancing. Even when the covariates used for identification coincide with the covariates used for balancing, our estimand generally differs from \citeauthor{caron2025triple}'s (\citeyear{caron2025triple}) CDATT because the two parameters use different reference populations and weighting schemes.

In addition, \citet{caron2025triple} does not develop a high-dimensional nuisance-estimation framework with cross-fitting, whereas our influence-function-based estimator accommodates flexible machine-learning methods for the nuisance components.

Taken together, these distinctions show that covariate adjustment and covariate balancing play different roles in group-level DiD comparisons. This paper fills this gap by introducing BGATT and DiBGATT as balanced group treatment-effect parameters and by developing doubly robust estimators that combine DiD adjustment with the reweighting needed to compare groups under a common covariate distribution.

\bigskip

% --------------------------------------------------
% Setup, Notation, and Definition of BGATT
% --------------------------------------------------
\section{Setup and Definition of BGATT}
\label{sec:setup}

This section lays out the basic two-period Difference-in-Differences framework, notation, and identifying assumptions. We then define group-specific ATTs, introduce the BGATT as the group effect evaluated under a common covariate distribution, and define its difference across groups (DiBGATT) as our main heterogeneity parameter.

\medskip

\subsection{Setup and notation}

Consider the classical difference-in-differences (DiD) setting with two time periods: a pre-treatment period ($t=0$) and a post-treatment period ($t=1$). In later sections, we extend our analysis to cross-sectional settings and to multiple time periods with staggered treatment adoption.

We observe panel data on $N$ independent units indexed by $i=1,\dots,N$. For each unit, we observe the random vector
\[
O_i = (Y_i, D_i, Z_i, X_i),
\]
where $Y_i = (Y_{i0}, Y_{i1})$ denotes outcomes over time. The treatment indicator $D_{it} \in \{0,1\}$ equals one if unit $i$ has received treatment by time $t$, and zero otherwise. We assume that no unit is treated in the pre-treatment period, so that $D_{i0}=0$ for all $i$, and define $D_i \equiv D_{i1}$ as the post-treatment treatment indicator.

The variable $Z_i \in \{0,1\}$ denotes a binary moderator measured prior to treatment, and $X_i$ is a vector of pre-treatment covariates. We partition $X_i=(W_i,V_i)$, where $W_i$ denotes the covariates whose distribution is balanced across values of $Z_i$, and $V_i$ contains additional covariates used for adjustment.

We adopt the potential outcomes framework (see, e.g., \cite{rubin1974estimating} and \cite{robins1986a}). Let $Y_{it}(0)$ denote the potential outcome for unit $i$ at time $t$ in the absence of treatment, and let $Y_{it}(1)$ denote the potential outcome under treatment. Under the Stable Unit Treatment Value Assumption (SUTVA) and consistency, the observed outcome satisfies
\[
Y_{it} = D_{it} Y_{it}(1) + (1-D_{it})Y_{it}(0), \quad t \in \{0,1\}.
\]

For completeness, we define potential values for the moderator and covariates, denoted $(Z_{i,t}(0), Z_{i,t}(1))$ and $(X_{i,t}(0), X_{i,t}(1))$, respectively, although our analysis conditions only on their pre-treatment realizations.
For notational convenience, we omit the subscript $i$ throughout the paper.

Due to the fundamental problem of causal inference, we can only observe either $Y_{t}(1)$ or $Y_{t}(0)$ for $t=1$, but never both. Consequently, the individual treatment effect on the treated, $Y_{1}(1) - Y_{1}(0)$, cannot be directly estimated. Instead, in the DiD framework, the primary parameter of interest is the average treatment effect on the treated (ATT), given by $\tau = E[Y_1(1) - Y_1(0) | D=1]$. In the conditional DiD setting, identification of this effect relies on the following assumptions: \medskip

\begin{assumption}[Conditional Parallel Trends.]\label{ass:pta} 
$E[Y_1(0)-Y_0(0)|D=1,X, Z]=E[Y_1(0)-Y_0(0)|D=0,X, Z]$ a.s.. 
\end{assumption}

\begin{assumption}[No anticipation assumption]
\label{ass:no_anticipation} 
$E[Y_0(0)|D=1,X,Z]=E[Y_0(1)|D=1,X,Z]$. 
\end{assumption}

\begin{assumption}[Strong Overlap]
\label{ass:overlap}
For each $z\in\{0,1\}$, there exists $\varepsilon>0$ such that $P(D=1)>\varepsilon$, $P(D=1\mid X,Z=z) \leq 1-\varepsilon$ almost surely, and $\varepsilon \leq P(Z=z\mid D=1,W) \leq 1-\varepsilon$ almost surely.
\end{assumption}

\begin{assumption}[Exogeneity]
\label{ass:exogeneity} 
$X(1)=X(0)=X$ and $Z(1)=Z(0)=Z$. 
\end{assumption}

Assumption \ref{ass:pta} imposes conditional parallel trends within cells defined jointly  by covariates $X$ and the moderator $Z$, requiring that the counterfactual trend for the treated would have equaled that of the untreated conditional on $(X, Z)$.   
Assumption \ref{ass:no_anticipation} concerns the pre-treatment outcome of the treated group, ensuring comparability necessary for identification.
Assumption \ref{ass:overlap} ensures that the treatment and control group have overlapping characteristics, preventing perfect separation between the two groups. 
Assumption \ref{ass:exogeneity} imposes the exogeneity of covariates $X$ and the moderator $Z$, meaning that they are measured prior to treatment and that they are not influenced by the treatment assignment $D$.

Throughout, we use subscripts to denote conditioning sets in probability densities and mass functions. In particular, we write $p_{A \mid B}(a) \equiv p(A = a \mid B)$,
so that, for example, $p_{X \mid D = 1, Z = z, W = w}(x)$
denotes the conditional density of \(X\) given \(D=1\), \(Z=z\), and \(W=w\). For discrete variables, we similarly write probability mass functions using subscripts, e.g.,
\( p_D(1) \equiv \Pr(D = 1) \).
Further define $m_1(X,z)=E[Y_1-Y_0|D=1,X,Z=z]$ and $m_0(X,z)=E[Y_1-Y_0|D=0,X,Z=z]$.

\bigskip

\subsection{Definition of BGATT and DiBGATT}
\label{sec:bgatt_definition}

In many empirical applications, researchers seek to examine how the treatment effect varies across different subpopulations rather than focusing solely on the overall treated group. This motivates the analysis of heterogeneity in the ATT with respect to some or all covariates $X$. A useful distinction in this context is between the individualized ATT (IATT) and the group ATT (GATT). The IATT captures the average effect on the treated for specific values of all covariates:
\begin{align*}
    \tau_1(x,z) = E[Y_1(1)-Y_1(0)|D=1,Z=z,X=x].
\end{align*}
The GATT gives the average effect on the treated across different covariate-defined groups:
\begin{align*}
    \gamma^G(z) &= E[Y_1(1)- Y_1(0)|D=1, Z=z] \\
    &= E[\tau_1(X,z)|D=1, Z=z]
\end{align*}

As noted by several authors, including \cite{bearth2024causal}, such parameters can be challenging to interpret due to differences in the distribution of the remaining covariates. Accounting for this heterogeneity is crucial for meaningful comparisons. 
To address this issue, we introduce a covariate-balanced group ATT (BGATT), which balances the distribution of $W$ across values of $Z$ using the treated population as the target distribution. Specifically, for $z\in\{0,1\}$,
\begin{align*}
    \gamma^B(z)
    &= E_{W\mid D=1}
    \left[E[Y_1(1)-Y_1(0)\mid D=1,Z=z,W]\right] \\
    &= E_{W\mid D=1}
    \left[ E[\tau_1(X,z)\mid D=1,Z=z,W] \right].
\end{align*}

A natural next step is to compare BGATTs across groups. For the binary moderator (Z), this comparison is given by the difference in BGATTs, the DiBGATT,
\begin{align*}
    \gamma^{\Delta B} =\gamma^B(1)-\gamma^B(0),
\end{align*}
which preserves the common covariate distribution used in the BGATT construction. DiBGATT is designed to answer a simple but central question: do treatment effects differ across groups once differences in covariate composition are held fixed? Standard comparisons of group-specific ATTs conflate treatment-response heterogeneity with differences in who is treated, making it unclear whether observed gaps reflect genuine effect heterogeneity or compositional differences. By equalizing covariate distributions across groups before comparing treatment effects, DiBGATT isolates heterogeneity in treatment responses rather than differences in baseline characteristics. This makes DiBGATT a transparent and policy-relevant estimand for assessing group-level heterogeneity in Difference-in-Differences designs. 

\bigskip

% --------------------------------------------------
% Identification, Decomposition, and Comparison to Existing DiD Estimands
% --------------------------------------------------
\section{Identification, Decomposition and Interpretation}
\label{sec:identification_decomposition}

This section introduces the main causal parameters studied in the paper and establishes their interpretation. We begin by defining the Balanced Group Average Treatment Effect on the Treated (BGATT), which evaluates group-specific treatment effects under a common distribution of balancing covariates. We then show how BGATT is identified. Building on this result, we define the difference in BGATTs (DiBGATT) and decompose the observed difference in group-specific ATTs into a balanced treatment-effect component and a compositional component. The section concludes by clarifying the interpretation of DiBGATT.

\medskip

\subsection{Identification}
\label{sec:identification}

Proposition \ref{prop:identification_bgatt} formalizes the BGATT parameter and establishes an identification result that will be essential for estimation and inference. \\

\begin{proposition}[Identification of BGATT] \label{prop:identification_bgatt}
Suppose Assumptions \ref{ass:pta}, \ref{ass:no_anticipation}, and \ref{ass:exogeneity} hold. Let the target distribution for balancing be \(F_{W\mid D=1}\). Then, for \(z\in\mathcal{Z}\), the balanced group average treatment effect on the treated is
\[\gamma^B(z)=E\!\left[E\!\left[Y_1(1)-Y_1(0)\mid D=1,Z=z,W\right]\mid D=1\right],\]
and is identified by
\[\gamma^B(z)=E\!\left[E\!\left[m_1(X,z)-m_0(X,z)\mid D=1,Z=z,W\right]\mid D=1\right],\]
where
\[m_d(x,z)=E[Y_1-Y_0\mid D=d,Z=z,X=x].\]
Equivalently,
\[m_d(x,z)=\mu_{d,1}(z,x)-\mu_{d,0}(z,x),\qquad\mu_{d,t}(z,x)=E[Y_t\mid D=d,Z=z,X=x].\]
\end{proposition}

\begin{proof}
See Appendix \ref{app:proof_identification}.
\end{proof} 

Because of linearity, it follows that the difference in BGATT (DiBGATT or $\Delta$BGATT) can be written as
\begin{align*}
    & \gamma^{\Delta B} = \gamma^B(1) - \gamma^B(0) \\
    &= E\!\left[ E\!\left[Y_1(1)-Y_1(0)\mid D=1,Z=1,W\right] - E\!\left[Y_1(1)-Y_1(0)\mid D=1,Z=0,W\right] \mid D=1 \right] \\
    &= E\!\left[E\!\left[m_1(1,X)-m_0(1,X)\mid D=1,Z=1,W\right] -E\!\left[m_1(0,X)-m_0(0,X)\mid D=1,Z=0,W\right]\mid D=1\right].
\end{align*}

\bigskip

\subsection{Decomposition}
\label{sec:decomposition}

We now derive a decomposition of $\gamma^{\Delta G} \equiv \gamma^{G}(1) - \gamma^{G}(0)$, the difference in group average treatment effects on the treated across $Z = 1$ and $Z = 0$.\footnote{A similar symmetric decomposition for cross-sectional data under unconfoundedness is derived by \cite{bearth2024causal}.} \\

\begin{proposition}[Decomposition of Difference in GATT] \label{prop:decomposition_bgatt}
$$\gamma^{\Delta G} = \gamma^{\Delta B} + \mathcal{C}_1 + \mathcal{C}_2,$$
\textit{where}
$$\gamma^{\Delta B} \;\equiv\; E\!\left[E[Y_1(1) - Y_1(0) \mid W, Z = 1] \mid D = 1\right] - E\!\left[E[Y_1(1) - Y_1(0) \mid W, Z = 0] \mid D = 1\right],$$
$$\mathcal{C}_1 \;\equiv\; E\!\left[E[Y_1(1) - Y_1(0) \mid W, Z = 1] \mid D = 1, Z = 1\right] - E\!\left[E[Y_1(1) - Y_1(0) \mid W, Z = 1] \mid D = 1\right],$$
$$\mathcal{C}_2 \;\equiv\; E\!\left[E[Y_1(1) - Y_1(0) \mid W, Z = 0] \mid D = 1\right] - E\!\left[E[Y_1(1) - Y_1(0) \mid W, Z = 0] \mid D = 1, Z = 0\right].$$ 
\end{proposition}

\begin{proof} 
Applying the law of iterated expectations to each $GATT(z)$ gives $\gamma^{\Delta G} = E[E[Y_1(1) - Y_1(0) \mid W, Z = 1] \mid D = 1, Z = 1] - E[E[Y_1(1) - Y_1(0) \mid W, Z = 0] \mid D = 1, Z = 0]$. Adding and subtracting $E[E[Y_1(1) - Y_1(0) \mid W, Z = 1] \mid D = 1]$ and $E[E[Y_1(1) - Y_1(0) \mid W, Z = 0] \mid D = 1]$, which average the group-specific conditional treatment effects over the marginal distribution of $W$ among the treated, and rearranging gives the result.
\end{proof}

The three components have distinct interpretations. $\gamma^{\Delta B}$ measures how much the treatment effect differs across groups for a treated individual drawn randomly from the marginal distribution of $W$, holding covariate composition fixed. The terms $\mathcal{C}_1$ and $\mathcal{C}_2$ are compositional effects: $\mathcal{C}_1$ captures how much the covariate distribution of group $Z = 1$ deviates from that of the treated population as a whole, and $\mathcal{C}_2$ the corresponding deviation for $Z = 0$. Both terms vanish if $W \perp Z \mid D = 1$, that is, if group membership is independent of covariates among the treated, or if treatment effects are homogeneous in $W$. Either condition alone is sufficient for $\gamma^{\Delta G} = \gamma^{\Delta B}$.\footnote{Our $\gamma^{\Delta G}$ is equivalent to \citeauthor{caron2025triple}'s (\citeyear{caron2025triple}) DATT.}

This decomposition is in the spirit of \cite{kitagawa1955components}, \cite{oaxaca1973male}, and \cite{blinder1973wage}, who decompose mean outcome gaps between two groups into a component attributable to differences in covariate distributions (the endowment effect) and a component attributable to differences in how covariates relate to outcomes (the returns effect). Here $\gamma^{\Delta B}$ plays the role of the returns effect and $\mathcal{C}_1 + \mathcal{C}_2$ the endowment effect, adapted to the potential outcomes setting.

A related two-term decomposition is proposed by \cite{caron2025triple}, who uses one group's covariate distribution as the reference, yielding a causal component and a treatment effect heterogeneity component.\footnote{A related decomposition in a continuous treatment DiD setting is derived in \cite{callaway2024difference}; see their Theorem 3.2(b).} As discussed in the Kitagawa-Oaxaca-Blinder literature, such two-term decompositions are asymmetric in the choice of reference group \citep{oaxaca1973male, oaxaca1994discrimination}. Proposition \ref{prop:identification_bgatt} resolves this by using the marginal distribution of $W$ of the treated population as a common reference point, yielding a decomposition that is symmetric across the two groups.

The advantage of using the marginal distribution as the reference is perhaps most transparent when $Z$ denotes gender. \citeauthor{caron2025triple}'s (\citeyear{caron2025triple}) decomposition asks, for instance, how much treatment effects differ by gender for a person with the covariate profile of an actual woman, meaning it conflates a genuine gender difference with the fact that women may have different observable characteristics than men. Proposition \ref{prop:identification_bgatt} instead asks how much effects differ for the average treated individual, holding covariate composition fixed across groups. The latter is the more natural quantity when the goal is to assess whether a treatment works differently for men and women per se, rather than as a by-product of compositional differences.

The decomposition in Proposition \ref{prop:decomposition_bgatt} also lends itself directly to estimation. Applying the law of total expectation to \(\mathcal{C}_1\) and \(\mathcal{C}_2\) over \(Z\mid D=1\), together with the identities
\begin{align*}
    & E[E[Y_1(1)-Y_1(0)\mid D=1,W,Z=z]\mid D=1] \\ 
    &= \sum_{z\in\{0,1\}} P(Z=z\mid D=1)E[E[Y_1(1)-Y_1(0)\mid D=1,W,Z=z]\mid D=1,Z=z],
\end{align*}
gives the equivalent representation
\begin{align*}
    \gamma^{\Delta G}
      &= \gamma^{\Delta B} \\
      & \quad + \frac{P(Z = 0 \mid D = 1)}{P(Z = 1 \mid D = 1)}
        \Bigl\{E\!\left[E[Y_1(1) - Y_1(0) \mid W, Z = 1] \mid D = 1\right] \\
      &\qquad - E\!\left[E[Y_1(1) - Y_1(0) \mid W, Z = 1] \mid D = 1, Z = 0\right]\Bigr\} \\
      & \quad - \frac{P(Z = 1 \mid D = 1)}{P(Z = 0 \mid D = 1)}
        \Bigl\{E\!\left[E[Y_1(1) - Y_1(0) \mid W, Z = 0] \mid D = 1\right] \\
      &\qquad - E\!\left[E[Y_1(1) - Y_1(0) \mid W, Z = 0] \mid D = 1, Z = 1\right]\Bigr\}.
\end{align*}

\bigskip

\subsection{Illustrative Example: Group Differences and Covariate Composition}

We illustrate the distinction between an unbalanced group treatment-effect difference and a balanced group treatment-effect difference using a simple example. Let \(Z\in\{0,1\}\) denote group status, with \(Z=1\) for women and \(Z=0\) for men, and let \(W\in\{h,\ell\}\) denote education, with \(h\) indicating high education and \(\ell\) indicating low education. Suppose that treatment effects vary by both gender and education. Among high-education units, the ATT is \(8\) for women and \(5\) for men. Among low-education units, the ATT is \(4\) for women and \(3\) for men. Thus, the gender treatment-effect gap equals \(8-5=3\) among high-education units and \(4-3=1\) among low-education units. At the same time, the treated covariate distributions differ across groups: among treated women, \(75\%\) are highly educated and \(25\%\) have low education, whereas among treated men, \(25\%\) are highly educated and \(75\%\) have low education. The target distribution used for balancing is the covariate distribution among all treated units, which in this example is \(P(W=h\mid D=1)=P(W=\ell\mid D=1)=0.5\).

The unbalanced difference in group-specific ATTs compares women and men under their own treated covariate distributions:
\[
\gamma^{\Delta G}
=
\left(0.75\cdot 8 + 0.25\cdot 4\right)
-
\left(0.25\cdot 5 + 0.75\cdot 3\right)
=
3.5.
\]
This contrast combines differences in treatment responses with differences in the distribution of education across treated women and men.

By contrast, the balanced group treatment-effect difference compares the two groups under the common covariate distribution, which assigns equal weight to high- and low-education units:
\[
\gamma^{\Delta B}
=
\left(0.5\cdot 8 + 0.5\cdot 4\right)
-
\left(0.5\cdot 5 + 0.5\cdot 3\right)
=
2.
\]
The difference between the two contrasts,
\[
\gamma^{\Delta G}-\gamma^{\Delta B}
=
3.5-2
=
1.5,
\]
is due to covariate composition. In this example, treated women are more likely to belong to the high-education group, where the gender treatment-effect gap is larger. The unbalanced group ATT difference therefore exceeds the treatment-effect gap obtained when women and men are compared under the same education distribution.

The example highlights the distinction between conditioning and balancing. Conditioning on education describes how the treatment-effect gap varies with \(W\). Balancing specifies how these conditional gaps are aggregated by choosing a common reference distribution for \(W\). Thus, DiBGATT captures the group treatment-effect difference that remains after equalizing covariate composition, while the difference between DiGATT and DiBGATT captures the contribution of compositional imbalance.

Figure~\ref{fig:DiBGATT_visualization} provides a graphical illustration of DiBGATT as the balanced difference in treatment effects across groups. The black solid lines show the observed evolution of untreated units, while the dashed line represents the counterfactual untreated outcome path for treated units under the parallel-trends assumption. The vertical distances between the treated potential outcomes and this counterfactual path correspond to treatment effects. Within the high-education group, the gender treatment-effect gap is \(8-5=3\), shown by the blue bracket. Within the low-education group, the corresponding gap is \(4-3=1\), shown by the orange bracket. DiBGATT averages these covariate-specific group gaps using the common target distribution of \(W|D=1\). In the numerical example, the target distribution assigns equal weight to high- and low-education units, so the balanced group difference can equivalently be computed as
\[
0.5\cdot 3 + 0.5\cdot 1 = 2.
\]
This is the same value obtained above by first constructing group-specific balanced ATTs and then taking their difference:
\[
\left(0.5\cdot 8+0.5\cdot 4\right)-\left(0.5\cdot 5+0.5\cdot 3\right)=2.
\]
Thus, the figure highlights that DiBGATT can be viewed either as the difference between balanced group-specific treatment effects or, equivalently, as the target-distribution-weighted average of covariate-specific treatment-effect gaps.

\begin{figure}[htbp]
    \centering
    \begin{tikzpicture}[font=\small, >=stealth, yscale=0.7]

%% -------------------------------------------------------
%% Coordinates  (x: time;  y: outcome)
%% -------------------------------------------------------
\def\xA{1.5}       % pre-treatment
\def\xB{6}         % opost-treatment
\def\xstart{0.4}   % axis start 
\def\xend{8.5}     % axis end 
\def\labelx{6.25}  % right-hand label start

% y-levels at t=0
\def\yCpre{0.6}
\def\yTpre{1.9}

% y-levels at t=1
\def\yCpost{2.6}
\def\yTcount{4.0}
\def\yRL{4.85}
\def\yRH{5.75}
\def\yBL{6.55}
\def\yBH{7.65}

% bracket x-position
\def\bkx{\xB + 3.85}

%% -------------------------------------------------------
%% Axes
%% -------------------------------------------------------
\draw[->, thick, mygray] (\xstart-0.4, 0) -- (\xend, 0);
\node[below, font=\small\sffamily] at (\xend-0.7, -0.15) {Time};
\draw[->, thick, mygray] (0, 0) -- (0, 7.8);
\node[rotate=90, anchor=south, font=\small\sffamily] at (-0.1, 6.4) {Outcome};

% ticks
\foreach \x/\yr in {\xA/0, \xB/1}{
  \draw[mygray] (\x, 0.07) -- (\x, -0.07);
  \node[below, font=\small] at (\x, -0.07) {\yr};
}

%% -------------------------------------------------------
%% Lines
%% -------------------------------------------------------

% Control (black solid)
\draw[black, thick]
  (\xA, \yCpre) -- (\xB, \yCpost);

% Counterfactual treated (black dashed)
\draw[black, thick, dashed]
  (\xA, \yTpre) -- (\xB, \yTcount);

% Red lines – low
\draw[myred, thick] (\xA, \yTpre) -- (\xB, \yRL);
\draw[myred, thick] (\xA, \yTpre) -- (\xB, \yRH);

% Blue lines – high education
\draw[myblue, thick] (\xA, \yTpre) -- (\xB, \yBL);
\draw[myblue, thick] (\xA, \yTpre) -- (\xB, \yBH);

%% -------------------------------------------------------
%% Endpoint dots
%% -------------------------------------------------------
\foreach \y in {\yCpre, \yTpre}
  \fill[black] (\xA, \y) circle (2.2pt);
\foreach \y in {\yCpost, \yTcount}
  \fill[black] (\xB, \y) circle (2.2pt);
\foreach \y in {\yRL, \yRH}
  \fill[myred] (\xB, \y) circle (2.2pt);
\foreach \y in {\yBL, \yBH}
  \fill[myblue] (\xB, \y) circle (2.2pt);

%% -------------------------------------------------------
%% Right-hand labels
%% -------------------------------------------------------
\node[left, font=\small] at (\xA-0.15, \yTpre) {$Y_0^T(0)$};
\node[left, font=\small] at (\xA-0.15, \yCpre) {$Y_0^C(0)$};

\node[right, myblue, font=\small] at (\labelx, \yBH)
  {$Y_1^T(1)\mid Z{=}\text{1},\, W{=}h$};
\node[right, myblue, font=\small] at (\labelx, \yBL)
  {$Y_1^T(1)\mid Z{=}\text{0},\, W{=}h$};
\node[right, myred, font=\small] at (\labelx, \yRH)
  {$Y_1^T(1)\mid Z{=}\text{1},\, W{=}\ell$};
\node[right, myred, font=\small] at (\labelx, \yRL)
  {$Y_1^T(1)\mid Z{=}\text{0},\, W{=}\ell$};
\node[right, font=\small] at (\labelx, \yTcount) {$Y_1^T(0)$};
\node[right, font=\small] at (\labelx, \yCpost)  {$Y_1^C(0)$};

%% -------------------------------------------------------
%% Brackets (DiBGATT parameters)
%% -------------------------------------------------------

% Blue bracket
\draw[myblue, thick] (\bkx, \yBL) -- (\bkx, \yBH);
\fill[myblue] (\bkx, \yBL) circle (2.2pt);
\fill[myblue] (\bkx, \yBH) circle (2.2pt);
\node[right, myblue, font=\small] at (\bkx+0.12, {0.5*(\yBL+\yBH)})
  {$\theta_h^{\mathit{\Delta GATT}}$};

% Red bracket
\draw[myred, thick] (\bkx, \yRL) -- (\bkx, \yRH);
\fill[myred] (\bkx, \yRL) circle (2.2pt);
\fill[myred] (\bkx, \yRH) circle (2.2pt);
\node[right, myred, font=\small] at (\bkx+0.12, {0.5*(\yRL+\yRH)})
  {$\theta_{\ell}^{\mathit{\Delta GATT}}$};

\end{tikzpicture}
    \caption{
    Illustration of DiBGATT in a \(2\times 2\) DiD setting with group status \(Z\in\{0,1\}\) and education \(W\in\{\ell,h\}\). The \textcolor{myblue}{blue} and \textcolor{myred}{orange} brackets show the within-education gender gaps in treatment effects for high- and low-education units, respectively. DiBGATT averages these covariate-specific gaps using the target distribution of \(W|D=1\), yielding \(\gamma^{\Delta B}\). The dashed line represents the counterfactual untreated outcome for treated units, \(Y_1^T(0)\), under parallel trends. For simplicity, untreated potential-outcome trends are assumed to be identical across groups.
    }
    \label{fig:DiBGATT_visualization}
\end{figure}
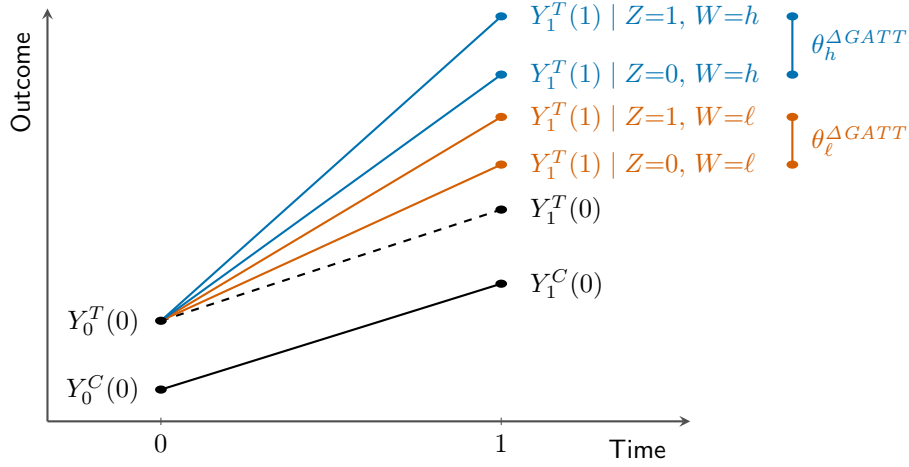

\bigskip

\subsection{Causal interpretation of the DiBGATT}
\label{sec:causal_interpretation}

While Proposition~\ref{prop:identification_bgatt} identifies the BGATT separately for each group \(Z\in\{0,1\}\), interpreting the difference
\[
\gamma^{\Delta B}=\gamma^B(1)-\gamma^B(0)
\]
as the causal effect of group status requires an additional assumption on the moderator \(Z\). Without this stronger interpretation, \(\gamma^{\Delta B}\) remains a well-defined and policy-relevant parameter. It compares treatment effects across groups after placing them on a common distribution of the observed balancing covariates \(W\). Thus, it answers whether treatment effects differ systematically across groups once differences in observed covariate composition have been removed.

This descriptive interpretation is often the relevant object for policy analysis. For example, when \(Z\) denotes gender, race, or another substantively important group characteristic, the researcher may be interested in whether an intervention has different effects across groups, without necessarily attributing the remaining difference causally to group status itself. A causal interpretation as the effect of \(Z\) is substantially stronger. It requires amongst other that, conditional on the covariates being held fixed, group status does not affect other determinants of the outcome or of treatment-effect heterogeneity. This exogeneity condition is unlikely to hold when \(Z\) shapes intermediate characteristics that are not fully captured by \(W\). For instance, if gender affects occupation, labor-market experience, health investments, or other confounders that in turn affect the outcome, then \(\gamma^{\Delta B}\) should not be interpreted as the causal effect of gender.

Accordingly, throughout the paper we interpret DiBGATT primarily as a balanced group contrast in treatment effects. 

\bigskip

% --------------------------------------------------
% Influence-function and Estimation
% --------------------------------------------------
\section{Influence Function and Estimation}
\label{sec:influence_function_estimation}

This section develops the influence-function representation for the Balanced Group Average Treatment Effect on the Treated (BGATT) and uses it to construct estimators suitable for flexible, high-dimensional settings. Building on this representation, we present estimation and inference procedures that are $\sqrt{n}$-consistent and asymptotically normal under standard regularity conditions.

\medskip

\subsection{Influence function for BGATT}
\label{sec:influence_function}

We next derive the influence function for the BGATT. An explicit influence-function representation enables the construction of estimators that remain valid when nuisance components—such as propensity scores or outcome regressions—are estimated flexibly using machine-learning methods, supporting $\sqrt{n}$-consistent estimation and asymptotically normal inference. In addition, the influence function provides a transparent characterization of how BGATT depends on the underlying data-generating process, clarifying the sources of identifying variation and the role of reweighting and covariate balance across groups. Theorem \ref{thm:if_bgatt} provides an explicit influence-function representation for the BGATT. 
\bigskip

\begin{theorem}[Influence Function for BGATT.] \label{thm:if_bgatt}
Under Assumptions \ref{ass:pta}--\ref{ass:exogeneity}, the influence function for the BGATT, $\phi_z$, is given by
\begin{align*}
    \phi_z &:= \mathbb{IF}(\gamma^B(z)) \\
    % &= \mathbb{IF}(\gamma^B_1(z)) - \mathbb{IF}(\gamma^B_0(z)) \\
    &= \omega_1(z;O) \cdot \left( Y_1 - Y_0 - m_0(X,z) \right) \\
    & \quad + \omega_2(z;O) \cdot \left( E[m_1(X,z)|D=1,Z=z,W] - E[m_0(X,z)|D=1,Z=z,W] \right) \\
    & \quad - \frac{D}{p_D(1)} \cdot \gamma^B(z)
\end{align*}
where
\begin{align*}
    &\omega_1(z;O) = \frac{\mathbf{1}_{\{Z=z\}}}{p_{Z \mid D=1, W}(z \mid w)\, p_D(1)} 
    \left(
    D - (1-D)\, \frac{p_{D \mid X, Z=z}(1 \mid x)}{p_{D \mid X, Z=z}(0 \mid x)}
    \right)
\end{align*}
and
\begin{align*}
    \omega_2(z;O) = \frac{D}{p_D(1)} - \frac{\mathbf{1}_{\{Z=z\}} \ D}{p_D(1) \ p_{Z \mid D=1, W}(z \mid w)}.
\end{align*}
\end{theorem}

\begin{proof} 
See Appendix \ref{app:derivation_if}.
\end{proof}

The influence-function representation in Theorem~\ref{thm:if_bgatt} has the same basic structure as doubly robust DiD scores for average treatment effects. The first term is a residualized DiD component: for units in group \(Z=z\), treated observations contribute their observed outcome change \(Y_1-Y_0\), while untreated observations are reweighted to recover the counterfactual untreated trend for the treated. The difference from standard doubly robust DiD scores is that BGATT is not evaluated under the group-specific covariate distribution, but under a common distribution of \(W\). This requires an additional balancing component that reweights group \(z\) to the marginal treated distribution of \(W\). Using Bayes' rule, the residual weight \(\omega_1(z;O)\) can be written as
\begin{align*}
    &\omega_1(z;O) = \underbrace{\frac{\mathbf{1}_{\{Z=z\}} \ D}{p_D(1) \ p_{Z \mid D=1, W}(z \mid w)}}_{\text{treated weight}} - \underbrace{\frac{\mathbf{1}_{\{Z=z\}} (1-D)}{p_{Z \mid D=0}(z)p_D(0)}}_{\text{untreated weight}} \ \underbrace{\frac{p_{X \mid D=1, Z=z}(x)}{p_{X \mid D=0, Z=z}(x)}}_{\text{density ratio across } D} \ \underbrace{\frac{p_{W \mid D=1}(w)}{p_{W \mid D=1, Z=z}(w)}}_{W\text{-balance factor}}.
\end{align*}
The first term is the treated contribution after reweighting group \(z\) to the target distribution of \(W\). The second term is the usual DiD control contribution, which aligns untreated units with treated units within group \(z\), multiplied by the same \(W\)-balance factor so that the counterfactual trend is also constructed for the balanced target population. Thus, relative to standard doubly robust DiD scores, \(\omega_1(z;O)\) adds the reweighting needed to compare treatment effects under a common covariate distribution.

The second weight, \(\omega_2(z;O)\), captures an additional feature of the BGATT estimand. Because BGATT averages the group-specific conditional treatment effect over the marginal treated distribution of \(W\), rather than over the distribution of \(W\) among treated units with \(Z=z\), the influence function must also account for the estimation of this target distribution. This is the role of \(\omega_2(z;O)\). It corrects the plug-in component
\[
E[m_1(X,z)\mid D=1,Z=z,W]
-
E[m_0(X,z)\mid D=1,Z=z,W]
\]
so that it is averaged over the correct reference population. Together, the \(\omega_1\) and \(\omega_2\) terms show that BGATT extends the usual doubly robust DiD score by combining the standard DiD adjustment with the influence-function terms required for covariate balancing.

When the covariates used for identification coincide with the covariates used for balancing, \(X=W\), the efficient influence function simplifies substantially. In this case, the nested conditioning on \(X\) within strata of \(W\) collapses, so the score no longer requires an additional projection from \(X\) to the balancing covariates.

\bigskip

% --------------------------------------------------
% Estimation and Double Machine Learning
% --------------------------------------------------
\subsection{Estimator Construction}
\label{sec:estimator}

This subsection describes the construction of an estimator for DiBGATT based on the influence-function representation established in the previous section. The proposed estimator follows the double machine learning (DML) paradigm, combining cross-fitting with flexible estimation of nuisance functions to obtain valid inference in high-dimensional settings.

\paragraph{Step 1: Cross-fitting and estimation of nuisance components.}
Let the sample be partitioned into $K$ mutually exclusive folds $\{\mathcal{I}_k\}_{k=1}^K$. For each fold $k$, we estimate the nuisance components using observations in the complement $\mathcal{I}_k^c$. Specifically, we obtain estimators
\[
\hat m_d^{[k]}(X,z) \approx \mathbb{E}[Y \mid D=d, X, Z=z], \quad d \in \{0,1\},
\]
the treatment propensity score $\hat p^{[k]}(D=1 \mid X,Z)$, the group assignment probability $\hat p^{[k]}(Z=z \mid D,W)$, and the conditional expectation
\[
\widehat{\mathbb{E}}^{[k]}\!\left[\hat m_d(X,z)\mid D=1, Z=z, W\right].
\]
These quantities may be estimated using generic machine learning methods, provided they satisfy the conditions stated in Section \ref{sec:asymptotics}.

\paragraph{Step 2: Estimation of DiBGATT.}
For each observation $i$, we evaluate the estimated influence function $\hat\phi(O_i)$ using nuisance estimates obtained from the fold $k(i)$ such that $i \in \mathcal{I}_{k(i)}$. The DiBGATT estimator $\hat\gamma^{\Delta B}$ is defined as the solution to the empirical moment condition
\[
\mathbb{E}_n\!\left[\hat\phi(O_i)\right] = 0.
\]
Given the linearity of the influence function in the target parameter, this equation admits an explicit closed-form solution.

\paragraph{Step 3: Asymptotic variance estimation and inference.}
Inference is conducted using the empirical variance of the estimated influence function. The asymptotic variance is estimated by
\[
\hat V = \mathbb{E}_n\!\left[\hat\phi(O_i)^2\right],
\]
yielding standard errors $\mathrm{SE}(\hat\gamma^{\Delta B}) = \sqrt{\hat V / n}$. Under the conditions stated in Section \ref{sec:asymptotics}, the resulting estimator is asymptotically normal, and $(1-\alpha)$ confidence intervals for DiBGATT are constructed as
\[
\hat\gamma^{\Delta B} \pm z_{1-\alpha/2}\,\mathrm{SE}(\hat\gamma^{\Delta B}),
\]
where $z_{1-\alpha/2}$ denotes the $(1-\alpha/2)$ quantile of the standard normal distribution.

\bigskip

% --------------------------------------------------
% Asymptotic Theory
% --------------------------------------------------
\subsection{Asymptotic Properties}
\label{sec:asymptotics}

We establish the large-sample properties of the BGATT estimator constructed in Section \ref{sec:estimator}. The key insight is that the doubly robust structure of $\phi_z$, combined with the cross-fitting procedure of Step 1, ensures that the bias from estimating the nuisance functions is second-order and therefore negligible at the $\sqrt{n}$ scale. We first state additional conditions required.

\begin{assumption}[Bounded Second Moments]
\label{ass:moments}
$E[\|\phi_z(O)\|^2] < \infty$ for each $z \in \mathcal{Z}$.
\end{assumption}

\begin{assumption}[Product Rate Condition]
\label{ass:product_rate}
For each $z \in \mathcal{Z}$, the nuisance estimators satisfy
\begin{align*}
    \|\hat{m}_d^{(k)} - m_d\|_{L_2} \cdot \|\hat{e}^{(k)} - e\|_{L_2} 
    &= o_p(n^{-1/2}), \\
    \|\hat{p}_Z^{(k)} - p_Z\|_{L_2} \cdot \|\hat{\Delta}_m^{(k)} 
    - \Delta_m\|_{L_2} &= o_p(n^{-1/2}),
\end{align*}
for each fold $k$, where $e(x,z) = P(D=1 \mid X=x, Z=z)$ and 
$\Delta_m(x,z,w) = E[m_1(X,z) - m_0(X,z) \mid D=1, Z=z, W=w]$.
\end{assumption}

\begin{assumption}[Stability]\label{ass:stability}
Consider a generic sample split into a training set $\mathcal I_2$ and a test
set $\mathcal I_1$, which are independent. For each $z \in \{0,1\}$, let
$\hat\delta_z(O_i)$ be the pseudo-outcome used in the second-step regression,
constructed from first-step nuisance estimates using only observations in
$\mathcal I_2$, and let $\delta_z(O_i)$ denote its population counterpart.
For $w$ in the support of $W$, define
\begin{align*}
    g_z(w)
        &= E\big[\delta_z(O_i)\mid D_i = 1,\, Z_i = z,\, W_i = w\big],\\
    \tilde g_z(w)
        &= E\big[\hat\delta_z(O_i)\mid D_i = 1,\, Z_i = z,\, W_i = w\big],\\
    \hat b_z(w)
        &= E\big[\hat\delta_z(O_i) - \delta_z(O_i)
                \mid D_i = 1,\, Z_i = z,\, W_i = w\big].
\end{align*}
Let $\hat g_z(w)$ denote the second-step regression estimator of $g_z(w)$
that regresses the pseudo-outcomes on $W$ using only the test sample
$\mathcal I_1$ (restricted to units with $D_i = 1, Z_i = z$). We assume that
$\hat g_z$ is \textit{stable} at $(z,w)$ with respect to some distance metric
$a$ in the sense that
\[
\frac{
    \hat g_z(w)
    - \tilde g_z(w)
    - E\big[\hat b_z(W_i)\mid D_i = 1,\, Z_i = z,\, W_i = w\big]
}{
    \big\{E\big[(\tilde g_z(w) - g_z(w))^2\big]\big\}^{1/2}
}
\xrightarrow{p} 0
\quad\text{whenever}\quad
a(\hat\delta_z,\delta_z) \xrightarrow{p} 0.
\]
\end{assumption}

Assumption \ref{ass:moments} is standard.
Assumption \ref{ass:product_rate} embodies a novel cross-double robustness property of the proposed estimator. Unlike standard doubly robust estimators, which require correct specification of one of two nuisance functions, consistency here holds if, for each of the following nuisance pairs, at least one component is correctly specified:
\begin{align*}
    (1)\quad &\bigl\{P(D \mid X, Z),\; 
               E[Y_1 - Y_0 \mid D=0, X, Z]\bigr\}, \\
    (2)\quad &\bigl\{P(Z \mid D=1, W),\; 
               E[m_1(X,Z) - m_0(X,Z) \mid D=1, Z, W]\bigr\},
\end{align*}
where $m_d(X,Z) = E[Y_1 - Y_0 \mid D=d, X, Z]$ for $d \in \{0,1\}$. 
Assumption \ref{ass:stability} says that the second-step regression estimator $\widehat{\mathbb{E}}[\cdot]$ used in the construction of our estimator is stable in the sense of \citet{Kennedy:2023}.

These conditions deliver the following asymptotic linear representation.\\

\begin{theorem}[Asymptotic Normality]
\label{thm:asymptotic}
Under Assumptions \ref{ass:pta}-\ref{ass:stability}, the BGATT estimator 
satisfies
\begin{align*}
    \sqrt{n}\bigl(\hat\gamma^B(z) - \gamma^B(z)\bigr) 
    = \frac{1}{\sqrt{n}}\sum_{i=1}^n \phi_z(O_i; P_0) + o_p(1)
    \;\xrightarrow{d}\; 
    \mathcal{N}\!\bigl(0,\, E_{P_0}[\phi_z(O;P_0)^2]\bigr),
\end{align*}
where $\phi_z(O;P_0)$ is the influence function characterized in Theorem \ref{thm:if_bgatt}. The estimator is $\sqrt{n}$-consistent, asymptotically normal, and semiparametrically efficient, with asymptotic variance equal to the semiparametric efficiency bound.
\end{theorem}

\begin{proof}
See Appendix \ref{app:asymptotic_prop}. The key steps are: (i) the von Mises expansion of $\gamma^B(\hat P) - \gamma^B(P_0)$ yields a leading linear term and a second-order remainder $R_2(\hat P, P_0)$; (ii) Neyman orthogonality of $\phi_z$ ensures $T_1 = o_p(n^{-1/2})$; (iii) the product rate condition (Assumption \ref{ass:product_rate}) ensures $R_2(\hat P, P_0) = o_p(n^{-1/2})$; and (iv) the central limit theorem applied to the leading term $S^* = (P_n - P)\{\phi(\cdot; P_0)\}$ delivers the result.
\end{proof}

\bigskip

% --------------------------------------------------
% Extensions
% --------------------------------------------------
\section{Extension to Multiple Time Periods and Staggered Adoption}
\label{sec:extension_multiple_periods}

We extend the framework to accommodate multiple time periods and staggered  treatment adoption.  Throughout this section we have repeated outcomes, use never-treated units as the comparison group, and abstract from treatment anticipation. Extension to repeated cross-sections, the not-yet-treated comparison group and to a known number $\delta > 0$ of anticipation periods are straightforward but omitted for brevity.

\medskip

\subsection{Setup and notation}

We follow the notation of \citet{callaway2021difference}. Let \(\mathcal{T}\) denote the number of time periods, indexed by \(t\in\{1,\ldots,\mathcal{T}\}\). Treatment adoption is irreversible, and \(G_i\) denotes the first period in which unit \(i\) is treated. Thus, \(G_i=g\) means that unit \(i\) first receives treatment in period \(g\), while \(G_i=\infty\) denotes units that are never treated.

\begin{assumption}[Staggered Treatment Adoption]
\label{ass:staggered}
$D_1 = 0$ almost surely, and $D_{t-1} = 1$ implies $D_t = 1$ almost surely for $t = 2, \ldots, \mathcal{T}$.
\end{assumption}

Let $G_i$ denote the period of first treatment, with $G_i = \infty$ for never-treated units. Define $G_{i,g} = \mathbf{1}\{G_i = g\}$, $C_i = \mathbf{1}\{G_i = \infty\}$, and let $\mathcal{G} \subseteq \{2,\ldots,\mathcal{T}\}$ denote the support of $G$. Let $Y_{i,t}(0)$ denote the potential outcome of unit $i$ at time $t$ absent treatment, and $Y_{i,t}(g)$ the potential outcome under first treatment in period $g$. Observed and potential outcomes satisfy
\begin{align}
    Y_{i,t} = Y_{i,t}(0) + \sum_{g=2}^{\mathcal{T}} 
    \bigl(Y_{i,t}(g) - Y_{i,t}(0)\bigr) \cdot G_{i,g}.
    \label{eq:observed_potential}
\end{align}
The generalized propensity score is defined as $p_g(X,Z) = P(G_g = 1 \mid G_g + C = 1, X, Z)$, the probability of first treatment in period $g$ conditional on covariates, moderator, and belonging either to group $g$ or the never-treated group.

The following assumptions extend Assumptions \ref{ass:pta}--\ref{ass:exogeneity} to the multiple time period setting.

\begin{assumption}[Conditional Parallel Trends]
\label{ass:pta_multi}
For each $g \in \mathcal{G}$ and $t \in \{2,\ldots,\mathcal{T}\}$,
\begin{align*}
    E\bigl[Y_t(0) - Y_{t-1}(0) \mid X, G_g = 1\bigr] = 
    E\bigl[Y_t(0) - Y_{t-1}(0) \mid X, C = 1\bigr] \quad \text{a.s.}
\end{align*}
\end{assumption}

\begin{assumption}[No Anticipation]
\label{ass:anticipation_multi}
For each $g \in \mathcal{G}$ and $t < g$, $E[Y_t(g) \mid G_g = 1, X, Z] = E[Y_t(0) \mid G_g = 1, X, Z]$ almost surely.
\end{assumption}

\begin{assumption}[Overlap]
\label{ass:overlap_multi}
For each $g \in \mathcal{G}$ and $t \in \{2,\ldots,\mathcal{T}\}$, there 
exists $\varepsilon > 0$ such that $P(G_g = 1) > \varepsilon$, 
$p_g(X,Z) < 1-\varepsilon$ almost surely, and 
$\varepsilon \leq p_{Z \mid G_g=1,W}(z \mid w) \leq 1-\varepsilon$ 
almost surely.
\end{assumption}

\begin{assumption}[Exogeneity]
\label{ass:exogeneity_multi}
$X_t(g) = X_t(0) = X_t$ and $Z_t(g) = Z_t(0) = Z_t$ almost surely, for all $g \in \mathcal{G}$ and $t \in \{1,\ldots,\mathcal{T}\}$.
\end{assumption}

\bigskip

\subsection{The Group-Time BGATT}

The parameter of interest is a group-time analog of the BGATT introduced in Section \ref{sec:identification},
\begin{align}
    \gamma^B_{g,t}(z) = E_W\!\Bigl[E\bigl[E[Y_t(g) - Y_t(0) 
    \mid G_g = 1, Z = z, X]\mid G_g = 1, Z = z, W\bigr]\Bigr],
    \label{eq:bgatt_multi}
\end{align}
capturing the average effect of first receiving treatment in period $g$, evaluated at time $t$, for units with moderator value $z$. \\

\begin{proposition}[Identification]
\label{prop:id_multi}
Under Assumptions \ref{ass:staggered}, \ref{ass:pta_multi}, \ref{ass:anticipation_multi}, and \ref{ass:exogeneity_multi} for each $g \in \mathcal{G}$ and $t \geq g$,
\begin{align*}
    \gamma^B_{g,t}(z) = E_W\!\Bigl[E\bigl[\mu_{g,t}^{nev}(z,X) 
    \mid G_g = 1, Z = z, W\bigr]\Bigr],
\end{align*}
where $\mu_{g,t}^{nev}(z,x) = E[Y_t - Y_{g-1} \mid G_g=1, Z=z, X=x] - E[Y_t - Y_{g-1} \mid C=1, Z=z, X=x]$.
\end{proposition}

\begin{proof}
The proof follows the same arguments as the proof of Proposition \ref{prop:identification_bgatt}, replacing $D$ with $G_g$, $C$ with the never-treated indicator, and $Y_1 - Y_0$ with the long difference $Y_t - Y_{g-1}$.
\end{proof}

\bigskip

\subsection{Aggregation}
\label{sec:aggregation}

The group-time parameters $\{\gamma^B_{g,t}(z)\}_{g,t}$ can be aggregated into summary measures using exactly the weighting schemes as stated in Table 1 of \cite{callaway2021difference}. For a set of weights $\{w(g,t)\}$ as defined therein, the aggregated BGATT at moderator value $z$ is
\begin{align}
    \Gamma^B(z) = \sum_{g \in \mathcal{G}}\sum_{t=2}^{\mathcal{T}} 
    w(g,t)\, \gamma^B_{g,t}(z),
    \label{eq:aggregated_bgatt}
\end{align}
which nests event-study, group-specific, calendar-time, and overall aggregations as special cases; see \cite{callaway2021difference} for the explicit forms of $w(g,t)$ in each case.

A key observation is that the weights $w(g,t)$ do not depend on $z$. Consequently, the DiBGATT, the difference between two BGATTs for different values of the moderator variable, satisfies
\begin{align*}
    \Gamma^B(1) - \Gamma^B(0) = \sum_{g \in \mathcal{G}}
    \sum_{t=2}^{\mathcal{T}} w(g,t)\,
    \bigl(\gamma^B_{g,t}(1) - \gamma^B_{g,t}(0)\bigr)
\end{align*}
by linearity, so that all asymptotic results derived for $\gamma^B_{g,t}(z)$ carry over immediately to aggregated parameters and their differences.

\bigskip

\subsection{Influence Function for Group-Time BGATT}

The influence function for the group-time BGATT is a direct extension of $\phi_z$ derived in Section \ref{sec:influence_function}, replacing $D$ with $G_g$, $p_D(1)$ with $\pi_g = P(G_g = 1)$, and $Y_1 - Y_0$ with the long difference $Y_t - Y_{g-1}$. Specifically, for each $(g,t,z) \in \mathcal{G} \times \{2,\ldots,\mathcal{T}\} \times \mathcal{Z}$,
\begin{align}
    \phi_{g,t,z} &= \omega_1^{g,t}(z;O) \cdot 
    \bigl(Y_t - Y_{g-1} - m_0^{g,t}(X,z)\bigr) + \omega_2^{g,t}(z;O) \cdot \Delta_m^{g,t}(z,W) 
    - \frac{G_g}{\pi_g}\,\gamma^B_{g,t}(z),
    \label{eq:if_multi}
\end{align}
where
\begin{align*}
    \omega_1^{g,t}(z;O) &= \frac{\mathbf{1}_{\{Z=z\}}}{p_{Z \mid G_g=1,W}
    (z \mid w)\,\pi_g} \left(G_g - (1-G_g)
    \frac{p_{G_g \mid X,Z=z}(1 \mid x)}{p_{G_g \mid X,Z=z}(0 \mid x)}\right), \\
    \omega_2^{g,t}(z;O) &= \frac{G_g}{\pi_g} - 
    \frac{\mathbf{1}_{\{Z=z\}}\,G_g}{\pi_g\,p_{Z \mid G_g=1,W}(z \mid w)},
\end{align*}
and $\Delta_m^{g,t}(z,W) = E[m_1^{g,t}(X,z)\mid G_g=1,Z=z,W] - E[m_0^{g,t}(X,z)\mid G_g=1,Z=z,W]$, with $m_j^{g,t}(x,z) = E[Y_t - Y_{g-1} \mid G_g = j, X=x, Z=z]$ for $j \in \{0,1\}$. 

\bigskip

\subsection{Asymptotic Properties}

Let $\boldsymbol{\phi}(O) = \{\phi_{g,t,z}(O)\}_{(g,t,z) \in \mathcal{I}}$ denote the vector of influence functions stacked across all $(g,t,z) \in \mathcal{I} = \mathcal{G} \times \{2,\ldots,\mathcal{T}\} \times \mathcal{Z}$, and let $\widehat{\mathbf{\Gamma}}^B$ and $\mathbf{\Gamma}^B$ denote the corresponding stacked vectors of estimated and true group-time BGATTs.

\begin{theorem}
\label{thm:asymptotic_multi}
Under Assumptions \ref{ass:product_rate}, \ref{ass:moments}, \ref{ass:staggered}, \ref{ass:pta_multi}, \ref{ass:anticipation_multi}, \ref{ass:overlap_multi}, and \ref{ass:exogeneity_multi}:

\noindent(i) For each $(g,t,z) \in \mathcal{I}$,
\begin{align*}
    \sqrt{n}\bigl(\widehat{\gamma}^B_{g,t}(z) - \gamma^B_{g,t}(z)\bigr) 
    = \frac{1}{\sqrt{n}}\sum_{i=1}^n \phi_{g,t,z}(O_i) + o_p(1).
\end{align*}

\noindent(ii) Jointly over $\mathcal{I}$,
\begin{align*}
    \sqrt{n}\bigl(\widehat{\mathbf{\Gamma}}^B - \mathbf{\Gamma}^B\bigr) 
    \xrightarrow{d} N(0, \Sigma),
\end{align*}
where $\Sigma = E[\boldsymbol{\phi}(O)\boldsymbol{\phi}(O)']$.
\end{theorem}

\begin{proof}
Analogous to the proof of Theorem \ref{thm:asymptotic}, replacing $\phi_z$ with $\phi_{g,t,z}$ and applying the multivariate Lindeberg--L\'{e}vy central limit theorem jointly over the finite index set $\mathcal{I}$.
\end{proof}

\noindent A consistent estimator of $\Sigma$ is given by 
$\hat{\Sigma} = \mathbb{E}_n[\hat{\boldsymbol{\phi}}(O) \hat{\boldsymbol{\phi}}(O)']$, where $\hat{\boldsymbol{\phi}}$ denotes the vector of estimated influence functions obtained by replacing all nuisance functions with their cross-fitted estimates. Pointwise confidence intervals follow immediately from Theorem \ref{thm:asymptotic_multi}. 

Pointwise confidence intervals for the aggregate parameters as discussed in Section \ref{sec:aggregation} follow by the delta method.
To see this, write a generic aggregate as \[ \theta=h(\mathbf{\Gamma}^B,\mathbf w), \qquad \widehat\theta=h(\widehat{\mathbf{\Gamma}}^B,\widehat{\mathbf w}), \] where $\mathbf w$ stacks the population aggregation weights. If \[ \sqrt n(\widehat{\mathbf w}-\mathbf w) = n^{-1/2}\sum_{i=1}^n \boldsymbol{\xi}^w(O_i)+o_p(1), \] and $h$ is continuously differentiable, then the aggregate influence function is \[ \ell_\theta(O) = \nabla_{\Gamma}h(\mathbf{\Gamma}^B,\mathbf w)' \boldsymbol{\phi}(O) + \nabla_w h(\mathbf{\Gamma}^B,\mathbf w)' \boldsymbol{\xi}^w(O). \] Hence \[ \sqrt n(\widehat\theta-\theta) = n^{-1/2}\sum_{i=1}^n \ell_\theta(O_i)+o_p(1), \qquad \widehat{SE}(\widehat\theta) = \left\{\mathbb E_n[\widehat\ell_\theta(O)^2]/n\right\}^{1/2}. \] For linear aggregates, the product of the estimated-weight and estimated-BGATT errors is second order, so uncertainty in both components enters through $\ell_\theta$.

Simultaneous inference, which is more appropriate given the multiple testing nature of the problem, is discussed in the next section (Section \ref{sec:multiplier_bootstrap}).

\bigskip

\subsection{Multiplier Bootstrap}
\label{sec:multiplier_bootstrap}

Similar to \cite{callaway2021difference}, we use a multiplier bootstrap to approximate the joint distribution of the group-time BGATT estimators and smooth aggregate parameters. Let $\{V_i\}_{i=1}^n$ be i.i.d. multipliers, independent of the data, with $E[V_i]=0$, $E[V_i^2]=1$, and finite moments of sufficiently high order. Define
\[
    \widehat{\mathbb G}_n^*
    =
    \frac{1}{\sqrt n}
    \sum_{i=1}^n V_i \widehat{\boldsymbol\phi}(O_i),
\]
where $\widehat{\boldsymbol\phi}$ stacks the estimated influence functions
$\widehat\phi_{g,t,z}$ over $(g,t,z)\in\mathcal I$. Provided that
$\mathbb E_n\|\widehat{\boldsymbol\phi}(O)-\boldsymbol\phi(O)\|^2=o_p(1)$, the multiplier-bootstrap argument of \cite{callaway2021difference} applied to Theorem \ref{thm:asymptotic_multi} yields
\[
    \widehat{\mathbb G}_n^*
    \rightsquigarrow_* N(0,\Sigma),
\]
where $\rightsquigarrow_*$ denotes weak convergence of the bootstrap law in probability, conditional on the original sample. Thus the bootstrap consistently approximates the joint limiting distribution of
$\sqrt n(\widehat{\mathbf\Gamma}^B-\mathbf\Gamma^B)$.

For a smooth aggregate parameter $\theta=h(\mathbf\Gamma^B,\mathbf w)$, the same construction is applied to the aggregate influence function $\ell_\theta$:
\[
    \widehat{\mathbb G}_{n,\theta}^*
    =
    \frac{1}{\sqrt n}
    \sum_{i=1}^n V_i \widehat\ell_\theta(O_i).
\]
Repeating this procedure $B$ times gives critical values for pointwise or simultaneous inference. For example, for a finite collection $\mathcal J$ of group-time parameters, simultaneous bands are based on the empirical $(1-\alpha)$-quantile of
\[
    \max_{j\in\mathcal J}
    \left|
    \frac{\widehat{\mathbb G}_{n,j}^{*,b}}
    {\widehat\sigma_j}
    \right|,
    \qquad
    \widehat\sigma_j^2=\mathbb E_n[\widehat\phi_j(O)^2].
\]
For aggregate parameters, replace $\widehat\phi_j$ and $\widehat\sigma_j$ with the corresponding $\widehat\ell_{\theta_j}$ and
$\{\mathbb E_n[\widehat\ell_{\theta_j}(O)^2]\}^{1/2}$. The resulting confidence bands are asymptotically valid uniformly over any finite collection of parameters considered.

\bigskip

% --------------------------------------------------
% Simulation Study
% --------------------------------------------------
\section{Simulation Study}
\label{sec:simulations}

We assess the finite-sample performance of the proposed DiBGATT estimator in a simulation study designed to capture the main challenges motivating our approach: covariate imbalance across groups, high-dimensional adjustment, and treatment-effect heterogeneity along the balancing variables. The section first describes the data-generating process, then introduces the simulation scenarios, and finally reports the resulting bias, dispersion, and coverage patterns.

\medskip

\subsection{DGP}

This subsection describes the data-generating process used in the simulation study. The design captures high-dimensional controls, group-level treatment-effect heterogeneity, and covariate imbalance.

Let $V_i \in \mathbb{R}^p$ denote a vector of high-dimensional control covariates,
\[
V_i \sim \mathcal{N}(0,I_p).
\]
The effect of $V_i$ on treatment assignment is governed by a sparse coefficient vector
\[
\beta_V = \frac{1}{s}\big(s, s-1, \dots, 1, 0, \dots, 0\big)^\top,
\]
where only the first $s$ entries are nonzero ($s=5$ in the baseline specification).

The group indicator is
\[
Z_i \sim \text{Bernoulli}(0.5),
\]
and the balancing covariates $W_i \in \mathbb{R}^q$ are generated conditional on $Z_i$ as
\[
W_i \mid Z_i \sim \mathcal{N}\big(Z_i \mathbf{1}_q,\, 4 I_q\big),
\]
so that the distribution of $W_i$ differs across groups. The corresponding coefficient vector is
\[
\delta = \frac{1}{\sum_{j=1}^q j}\,(q, q-1, \dots, 1)^\top,
\]
assigning decreasing weights to the components of $W_i$.

Treatment assignment follows a logistic model,
\[
\Pr(D_i = 1 \mid V_i, W_i, Z_i)
= \sigma\big(V_i^\top \beta_V + W_i^\top \delta + Z_i - 1\big),
\]
where $\sigma(x) = (1+e^{-x})^{-1}$, inducing selection on both $V_i$ and $W_i$.

Individual treatment effects are given by a function $\theta(Z_i,W_i)$. We consider two specifications:
\[
\theta_{\text{add}}(Z_i,W_i)
= 3(1+Z_i) + W_i^\top \delta,
\]
\[
\theta_{\text{int}}(Z_i,W_i)
= 3(1+Z_i) + W_i^\top \delta + Z_i\, W_i^\top \delta,
\]
where the latter introduces interactions between group status and $W_i$.

Outcomes follow a two-period DiD structure. Untreated potential outcomes evolve according to a common time trend,
\[
Y_{i,01} = Y_{i,00} + 1 + \varepsilon_{i,10}, \qquad
\varepsilon_{i,10} \sim \mathcal{N}(0,0.1^2),
\]
while treated potential outcomes in period 1 are
\[
Y_{i,11} = Y_{i,01} + \theta(Z_i,W_i) + \varepsilon_{i,11}, \qquad
\varepsilon_{i,11} \sim \mathcal{N}(0,0.1^2).
\]
Observed outcomes are
\[
Y_{i0} = Y_{i,00}, \qquad
Y_{i1} = (1-D_i) Y_{i,01} + D_i Y_{i,11}.
\]

This DGP satisfies conditional parallel trends by construction.

\bigskip

\subsection{Simulation scenarios}
\label{subsec:scenarios}

We report simulation results for two high-dimensional designs with multidimensional covariate balancing. In all reported designs, the balancing covariate has four components, $\dim(W)=4$, and the additional control vector is high-dimensional with $p=300$.

The first design uses the additive treatment-effect specification $\theta_{\mathrm{add}}(Z,W)$ and fixes the sample size at $n=1600$. This design isolates the role of nuisance estimation in a high-dimensional setting. We compare linear nuisance models, estimated by least squares and logistic regression, with cross-validated $\ell_1$-penalized nuisance models. Results are based on 5000 Monte Carlo replications.

The second design uses the interactive treatment-effect specification $\theta_{\mathrm{int}}(Z,W)$, so that treatment effects vary with both group status and the balancing covariates. We estimate nuisance functions using cross-validated $\ell_1$-penalized methods and compare performance at $n=400$ and $n=1600$. This design assesses the behavior of DiBGATT in a more demanding setting with high-dimensional controls, covariate imbalance, and treatment-effect heterogeneity involving $W$. For each design, we report the sampling distribution of the DiBGATT estimator, together with the true value, empirical RMSE, and coverage of nominal 95\% confidence intervals.

\bigskip

\subsection{Simulation results}
\label{subsec:simulation_results}

Figure \ref{fig:sim_simple_nuisance} compares nuisance-estimation methods in the high-dimensional additive design. Linear nuisance models display greater dispersion and finite-sample distortion, whereas cross-validated lasso produces a sampling distribution that is more tightly centered around the true value. This comparison illustrates the value of regularization even when treatment-effect heterogeneity is additive.

Figure \ref{fig:sim_interactive_results} reports results for the interactive design. With $n=400$, the sampling distribution is more dispersed, reflecting the difficulty of estimating treatment-effect heterogeneity involving $W$ in the presence of high-dimensional controls. Increasing the sample size to $n=1600$ substantially tightens the distribution and improves coverage. Overall, the simulations support the main theoretical message: DiBGATT delivers reliable inference in high-dimensional DiD settings when nuisance components are estimated with regularized methods, while performance improves as the sample size increases.

\begin{figure}[htbp]
    \centering

    \begin{subfigure}{0.48\textwidth}
        \centering
        \includegraphics[width=\linewidth]{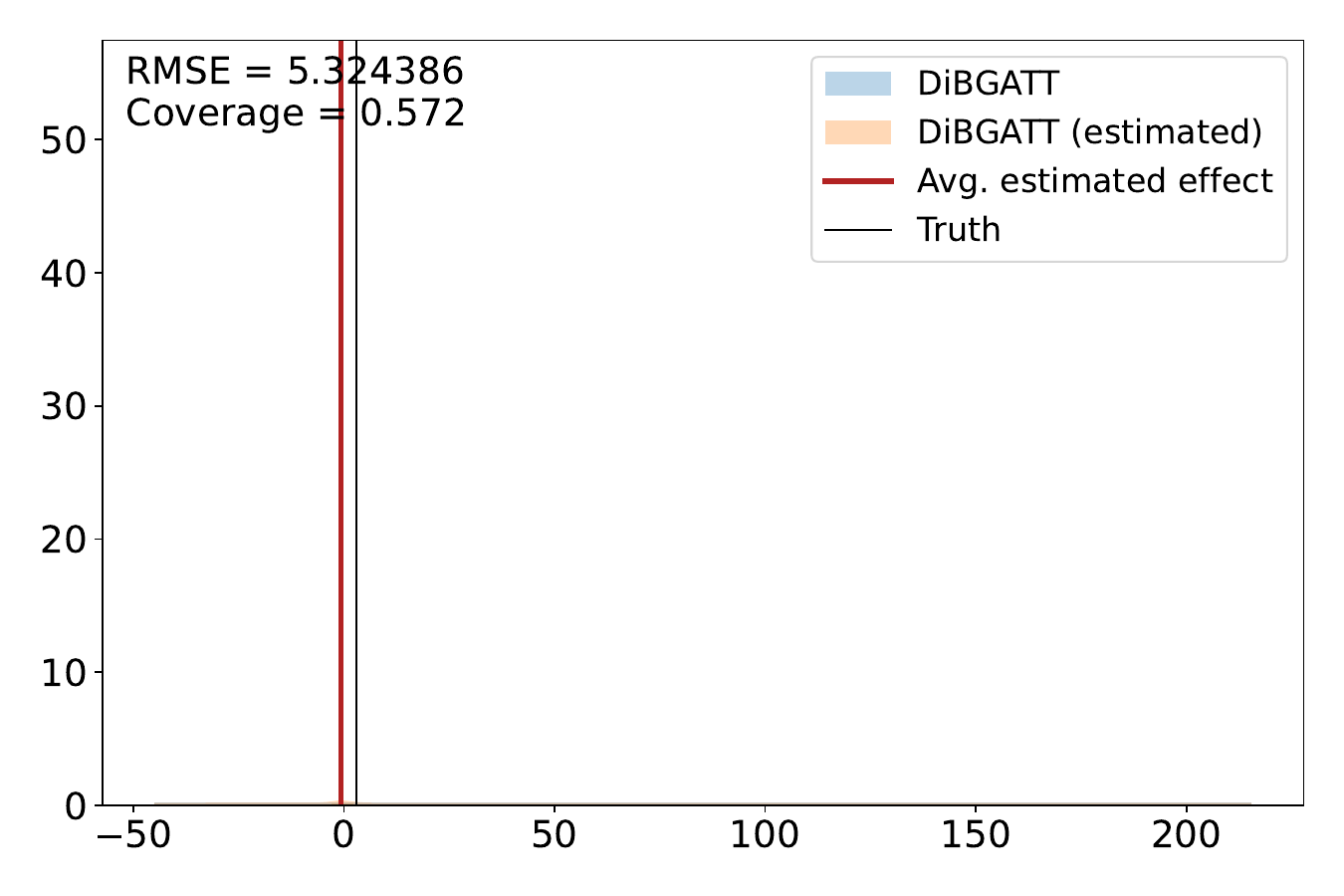}
        \caption{Simple, $\dim(W)=4$, $p=300$, $n=1600$, linear}
        \label{fig:sim_simple_linear}
    \end{subfigure}\hfill
    \begin{subfigure}{0.48\textwidth}
        \centering
        \includegraphics[width=\linewidth]{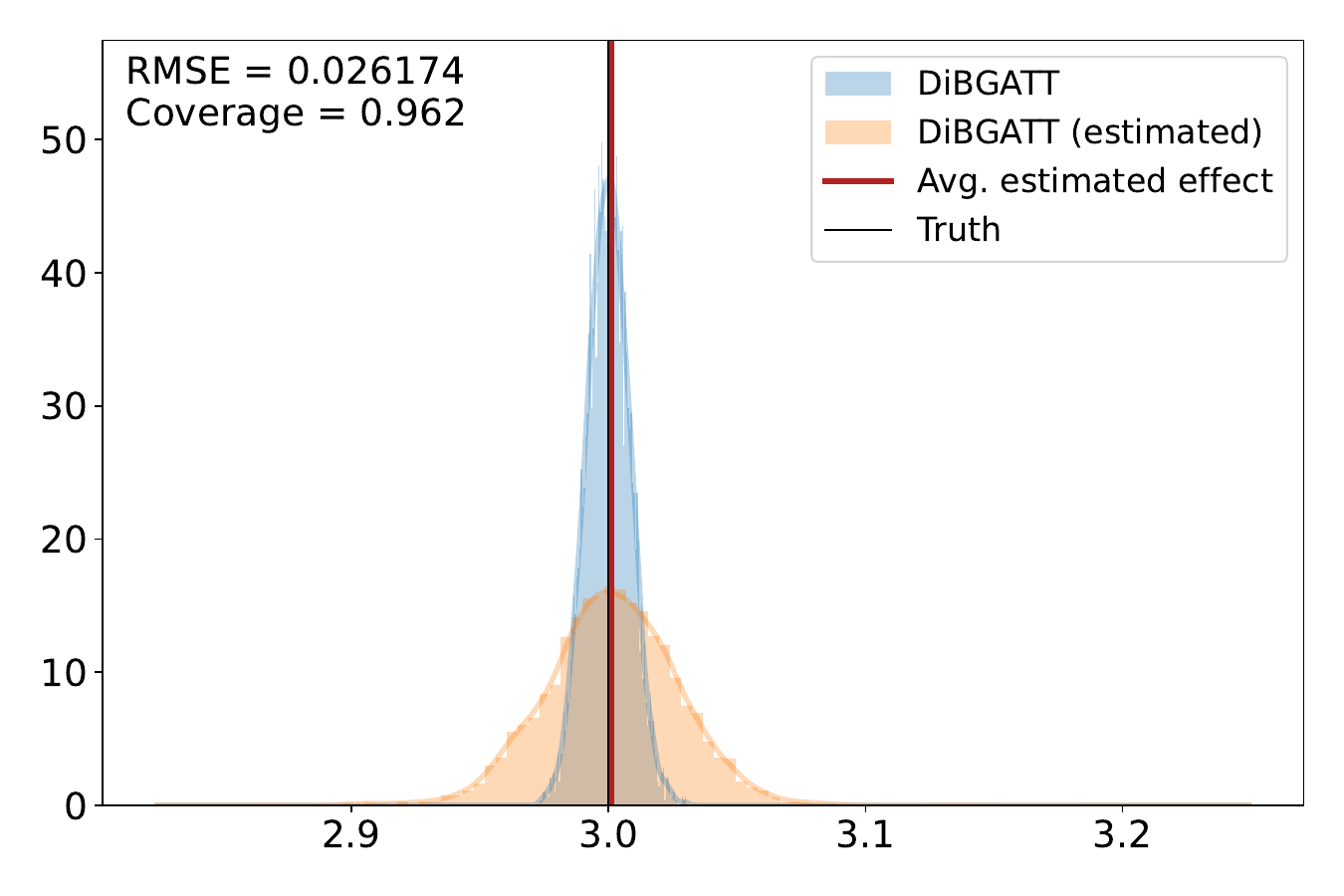}
        \caption{Simple, $\dim(W)=4$, $p=300$, $n=1600$, lasso}
        \label{fig:sim_simple_lasso}
    \end{subfigure}

    \caption{Sampling distributions of DiBGATT in the high-dimensional additive treatment-effect design, based on 5000 simulations.}
    \label{fig:sim_simple_nuisance}
\end{figure}

\begin{figure}[htbp]
    \centering

    \begin{subfigure}{0.48\textwidth}
        \centering
        \includegraphics[width=\linewidth]{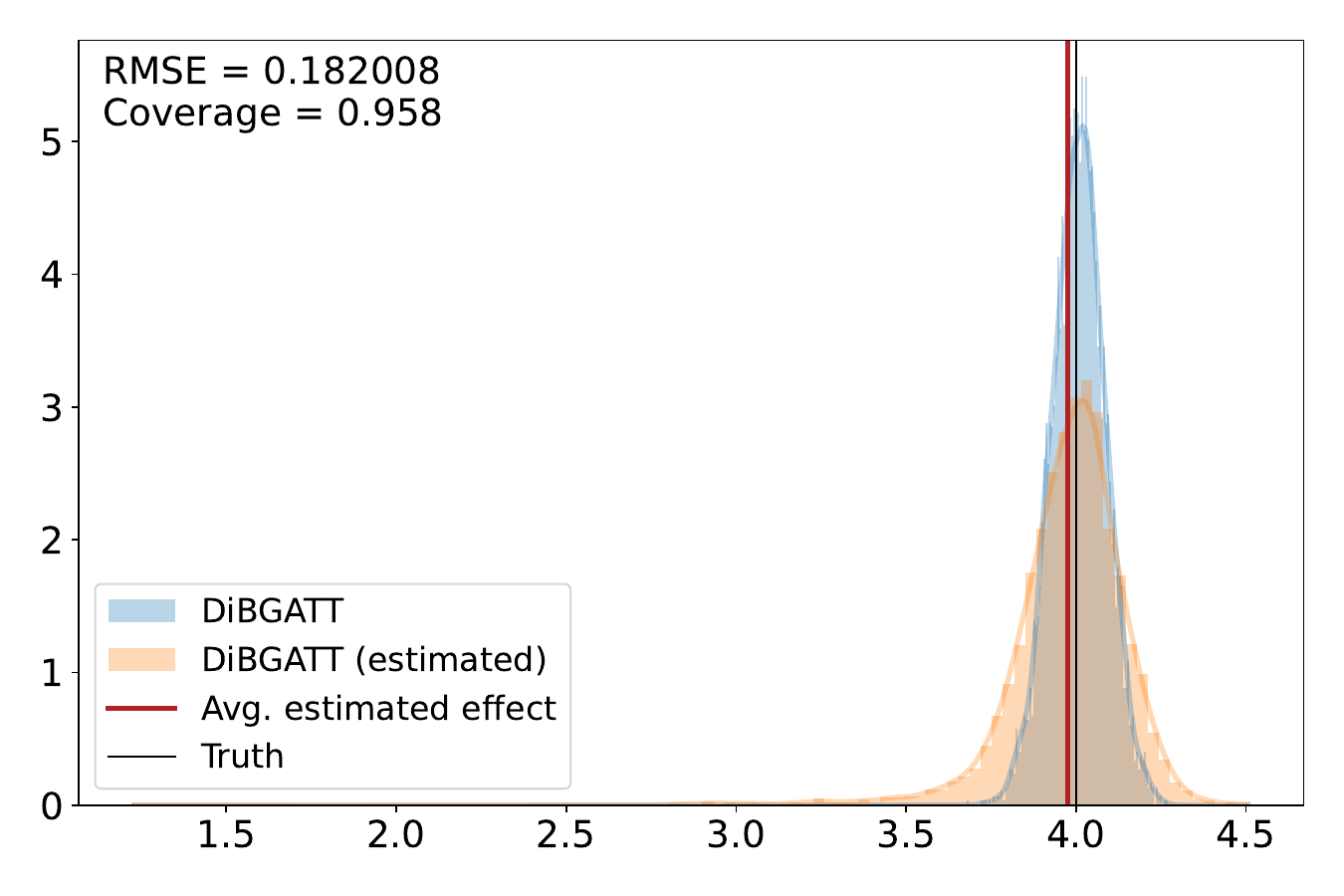}
        \caption{Interactive, $\dim(W)=4$, $p=300$, $n=400$}
        \label{fig:sim_interactive_n400}
    \end{subfigure}\hfill
    \begin{subfigure}{0.48\textwidth}
        \centering
        \includegraphics[width=\linewidth]{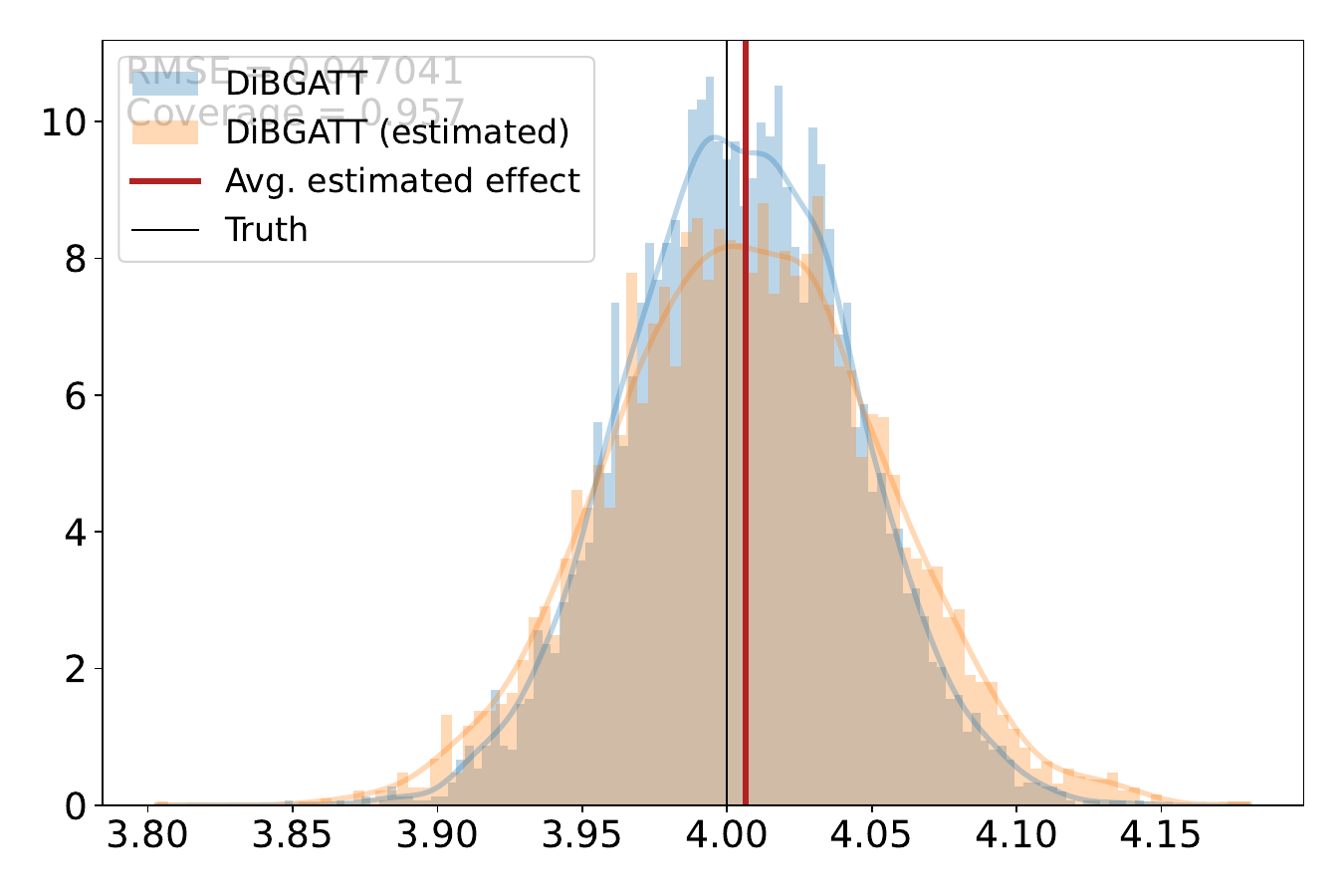}
        \caption{Interactive, $\dim(W)=4$, $p=300$, $n=1600$}
        \label{fig:sim_interactive_n1600}
    \end{subfigure}

    \caption{Sampling distributions of DiBGATT in the high-dimensional interactive treatment-effect design using lasso nuisance estimation.}
    \label{fig:sim_interactive_results}
\end{figure}

\bigskip

% --------------------------------------------------
% Empirical Applications
% --------------------------------------------------
\section{Empirical Application}
\label{sec:application}

\textit{Work in progress.}

\bigskip

% --------------------------------------------------
% Conclusion
% --------------------------------------------------
\section{Conclusion}
\label{sec:conclusion}

This paper introduces the Balanced Group Average Treatment Effect on the Treated (BGATT) and its difference across groups, DiBGATT, to study treatment effect heterogeneity in Difference-in-Differences designs. By equalizing covariate distributions across groups, the proposed framework isolates heterogeneity in treatment responses from compositional differences and enables direct, well-powered inference on group-level effects. We establish identification under standard parallel-trends assumptions, derive an efficient influence-function representation, and develop semiparametric estimation and inference procedures suitable for high-dimensional settings. Nuisance models can be estimated with flexible machine learning methods since our representation is cross-doubly-robust.

% --------------------------------------------------
% References
% --------------------------------------------------
\newpage
\bibliography{references}

% --------------------------------------------------
% Appendices
% --------------------------------------------------
\newpage
\begin{appendices}

\allowdisplaybreaks

\section{Proofs}
\label{app:proofs}

\subsection{Proof of Proposition \ref{prop:identification_bgatt}}
\label{app:proof_identification}

\textbf{Identification through outcome regression} 

\noindent Define $\mu_{d,t}(z,x)=E[Y_t|D=d,Z=z,X=x]$. Under assumption of parallel trends and using that $Y_1=Y_1(1)\cdot D + Y_1(0) \cdot (1-D)$ (SUTVA), implying that $Y_1(1)=Y_1$ if $D=1$ and $Y_1(0)=Y_1$ if $D=0$:
\begin{align*}
    & \gamma^B(z) \\
    &= E[E[Y_1(1)- Y_1(0)|D=1,Z=z,W] | D = 1] \\
    &= E[E[E[Y_1(1)- Y_1(0)|D = 1, X,Z]|D=1, Z=z,W]| D = 1] \\
    &= E[E[E[Y_1(1)|D = 1, X,Z]- E[Y_1(0)|D = 1, X,Z]|D=1,Z=z,W]| D = 1] \\
    &= E[E[E[Y_1(1)|D = 1, X,Z] - E[Y_0(0)|D = 1, X,Z] \\ 
    & \quad - E[Y_1(0)|D=0,X,Z] - E[Y_0(0)|D=0,X,Z]|D=1,Z=z,W]| D = 1] \\
    &= E[E[E[Y_1|D=1,X,Z] - E[Y_0|D = 1, X,Z]] \\ 
    & \quad - E[Y_1|D=0,X,Z] - E[Y_0|D = 1, X,Z]|D=1,Z=z,W]| D = 1] \\
    &= E[E[E[Y_1 - Y_0|D = 1, X,Z] - E[Y_1 - Y_0|D=0,X,Z] |D=1,Z=z,W]| D = 1] \\
    &= E[E[m_1(x,z)- m_0(x,z)|D=1,Z=z,W]| D = 1]
\end{align*}
The first equality comes from the law of iterated expectations. The third equality from the common trends assumption, the fourth equality from the no anticipation assumption and SUTVA.

\bigskip

\subsection{Proof of Theorem \ref{thm:if_bgatt}}
\label{app:derivation_if}

We start by defining the target estimand under a discrete data-generating process to simplify exposition. We then leverage the interpretation of the influence function as a functional derivative to derive the influence function for our estimand by combining known influence functions of its building blocks\footnote{This approach is based on strategy 2 of \cite{kennedy2024semiparametric}.}.
For group $Z=z$, the target estimand is the DiBGATT, defined as 
\begin{align*}
    \gamma^B(z) = E \left[ E \left[ E \left[ Y_1 - Y_0 | D=1, X, z \right] - E \left[ Y_1 - Y_0 | D=0, X, z \right] | D=1 , Z=z , W \right] | D=1 \right]
\end{align*}

We will first derive influence function for the BGATT, $E \left[ E \left[m_1(X,z) | D=1 , Z=z , W \right] | D=1 \right]$ denoted $\gamma^B_1(z)$, then for the BGATT $E \left[ E \left[m_1(X,z) | D=1 , Z=z , W \right] | D=1 \right]$ denoted $\gamma^B_0(z)$, since their difference gives the DiBGATT $\gamma^B(z) = \gamma^B_1(z)-\gamma^B_0(z)$ and due to linearity it follows that the DiBGATT's influence function simply is $IF(\gamma^B(z)) = IF(\gamma^B_1(z))-IF(\gamma^B_0(z))$.

\bigskip

\subsubsection{Derivation of $\gamma^B_1(z)$}

$\gamma^B_1(z)$ can be written as: 
\begin{align*}
    \gamma^B_1(z) = \int \left[ \int m_1(x, z) \, dP_{X \mid D=1, Z=z, W=w}(x) \right]
\, dP_{W \mid D=1}(w).
\end{align*}
We start by assuming that the data are discrete to simplify the exposition. Under discreteness, expressions involving sums of indicator functions collapse—for example, $\sum_x \mathbf{1}_{\{X = x\}} \, p(x) = p(X),$
which streamlines the derivations. The influence functions derived under this assumption are the same as those for continuous data:
\begin{align*}
    \gamma^B_1(z) = \sum_{w} \sum_{x} m_1(x,z)\, p_{X\mid D=1,Z=z,W=w}(x)\, p_{W\mid D=1}(w).
\end{align*}
Now use product rule for derivatives, $\mathbb{IF}(\gamma_1\gamma_2)=\mathbb{IF}(\gamma_1)\gamma_2+\gamma_1\mathbb{IF}(\gamma_2)$:
\begin{align*}
    & \mathbb{IF}(\gamma^B_1(z)) \\ &= \sum_w \sum_x \mathbb{IF}(m_1(x,z)) \ p_{X \mid D = 1, Z = z, W = w}(x) \ p_{W \mid D=1}(w) \\
    & \quad + \sum_w \sum_x m_1(x,z) \ \mathbb{IF}(p_{X \mid D = 1, Z = z, W = w}(x)) \ p_{W \mid D=1}(w) \\
    & \quad + \sum_w \sum_x m_1(x,z) \ p_{X \mid D = 1, Z = z, W = w}(x) \ \mathbb{IF}(p_{W \mid D=1}(w)) \\
    & = \frac{D}{p_D(1)} \Biggl\{
\underbrace{\sum_w \sum_x \frac{\mathbf{1}_{\{X=x,Z=z\}}}{p_{X,Z \mid D=1}(x,z)} (Y_1-Y_0-m_1(x,z)) \, p_{X \mid D=1, Z=z, W=w}(x) \ p_{W \mid D=1}(w)}_{\text{Term (a)}} \\
&\quad + \underbrace{\sum_w \sum_x m_1(x,z) \, \frac{\mathbf{1}_{\{Z=z,W=w\}}}{p_{Z,W \mid D=1}(z,w)} \big(\mathbf{1}_{\{X=x\}}-p_{X \mid D=1, Z=z, W=w}(x)\big) \, p_{W \mid D=1}(w)}_{\text{Term (b)}} \\
&\quad + \underbrace{\sum_w \sum_x m_1(x,z) \, p_{X \mid D=1, Z=z, W=w}(x) \ \big(\mathbf{1}_{\{W=w\}}-p_{W \mid D=1}(w)\big)}_{\text{Term (c)}}
\Biggr\}.
\end{align*}

\begin{align*}
\mathbb{IF}(\gamma^B_1(z))
&= \sum_w \sum_x \mathbb{IF}\!\left(m_1(x,z)\right)\,
p_{X\mid D=1,Z=z,W=w}(x)\,
p_{W\mid D=1}(w) \\
&\quad + \sum_w \sum_x
m_1(x,z)\,
\mathbb{IF}\!\left(p_{X\mid D=1,Z=z,W=w}(x)\right)\,
p_{W\mid D=1}(w) \\
&\quad + \sum_w \sum_x
m_1(x,z)\,
p_{X\mid D=1,Z=z,W=w}(x)\,
\mathbb{IF}\!\left(p_{W\mid D=1}(w)\right) \\
&= \frac{D}{p_D(1)} \Biggl\{
\underbrace{
\sum_w \sum_x
\frac{\mathbf{1}_{\{X=x,Z=z\}}}{p_{X,Z\mid D=1}(x,z)}
\big( Y_1-Y_0-m_1(x,z) \big)\,
p_{X\mid D=1,Z=z,W=w}(x)\,
p_{W\mid D=1}(w)
}_{\text{Term (a)}} \\
&\quad + \underbrace{
\sum_w \sum_x
m_1(x,z)\,
\frac{\mathbf{1}_{\{Z=z,W=w\}}}{p_{Z,W\mid D=1}(z,w)}
\Big(
\mathbf{1}_{\{X=x\}} - p_{X\mid D=1,Z=z,W=w}(x)
\Big)\,
p_{W\mid D=1}(w)
}_{\text{Term (b)}} \\
&\quad + \underbrace{
\sum_w \sum_x
m_1(x,z)\,
p_{X\mid D=1,Z=z,W=w}(x)\,
\Big(
\mathbf{1}_{\{W=w\}} - p_{W\mid D=1}(w)
\Big)
}_{\text{Term (c)}}
\Biggr\}.
\end{align*}

We now analyze each term separately.

\noindent \textbf{Term (a)}:
\begin{align*}
    & \sum_w \sum_x \frac{\mathbf{1}_{\{X=x,Z=z\}}}{p_{X,Z \mid D=1}(x,z)} (Y_1-Y_0-m_1(x,z)) \ p_{X \mid D=1, Z=z, W=w}(x) \ p_{W \mid D=1}(w) \\
    & = \sum_w \frac{\mathbf{1}_{\{Z=z\}}}{p_{X,Z \mid D=1}(x,z)} (Y_1-Y_0-m_1(X,z)) \ p_{X \mid D=1, Z=z, W=w}(x) \ p_{W \mid D=1}(w) \\
    & = \frac{\mathbf{1}_{\{Z=z\}}}{p_{X,Z \mid D=1}(x,z)} (Y_1-Y_0-m_1(X,z)) \sum_w  \ p_{X \mid D=1, Z=z, W=w}(x) \ p_{W \mid D=1}(w) \\
    & = \frac{\mathbf{1}_{\{Z=z\}}}{p_{V,W,Z \mid D=1}(v,w,z)} (Y_1-Y_0-m_1(X,z)) \sum_w \mathbf{1}_{\{W=w\}} \ p_{V \mid D=1, Z=z, W=w}(v) \ p_{W \mid D=1}(w)  \\
    & = \frac{\mathbf{1}_{\{Z=z\}} \ p_{V \mid D=1, Z=z, W=w}(v) \ p_{W \mid D=1}(w)}{p_{V \mid D=1, W, Z=z}(v \mid w) \ p_{W \mid D=1, Z=z}(w) \ p_{Z \mid D=1}(z)} (Y_1-Y_0-m_1(X,z))   \\
    & = \frac{\mathbf{1}_{\{Z=z\}} \ p_{W \mid D=1}(w)}{p_{W \mid D=1, Z=z}(w) \ p_{Z \mid D=1}(z)} (Y_1-Y_0-m_1(X,z)) \\
    & = \frac{\mathbf{1}_{\{Z=z\}}}{p_{Z \mid D=1, W}(z \mid w)} (Y_1-Y_0-m_1(X,z))
\end{align*}
Using for the third equality that $X=(V,W)$ and that
\begin{align*}
    p_{X \mid D=1, Z=z, W=w}(x) = p_{V,W \mid D=1, Z=z, W=w}(v,w)
    = \left\{
    \begin{aligned}
    &p_{V \mid D=1, Z=z, W}(v \mid w)  &&\text{if } w = W,\\[4pt]
    &0                      &&\text{if } w \neq W,
    \end{aligned}
    \right.
\end{align*}
and for the last equality that $p_{Z,W \mid D=1}(z,w)=p_{W \mid D=1}(w) \ p_{Z \mid D=1, W}(z \mid w)$.

\noindent \textbf{Term (b)}:
\begin{align*}
    & \sum_w \sum_x m_1(x,z) \, \frac{\mathbf{1}_{\{Z=z,W=w\}}}{p_{Z,W \mid D=1}(z,w)} \big(\mathbf{1}_{\{X=x\}}-p_{X \mid D=1, Z=z, W=w}(x)\big) \, p_{W \mid D=1}(w) \\
    & = m_1(X,z) \frac{\mathbf{1}_{\{Z=z\}}}{p_{Z,W \mid D=1}(z,w)} \ p_{W \mid D=1}(w) \\
    & \quad - \sum_x m_1(x,z) \frac{\mathbf{1}_{\{Z=z\}}}{p_{Z,W \mid D=1}(z,w)} \ p_{Z,W \mid D=1}(z,w) \ p_{W \mid D=1}(w) \\
    & = \frac{\mathbf{1}_{\{Z=z\}} \ p_{W \mid D=1}(w)}{p_{Z,W \mid D=1}(z,w)} \ m_1(X,z) - \frac{\mathbf{1}_{\{Z=z\}} \ p_{W \mid D=1}(w)}{p_{Z,W \mid D=1}(z,w)} \sum_x m_1(x,z) \ p_{Z,W \mid D=1}(z,w) \\
    & = \frac{\mathbf{1}_{\{Z=z\}}}{p_{Z \mid D=1, W}(z \mid w)} \cdot \left( m_1(X,z) - E[m_1(X,z)|D=1,Z=z,W] \right)
\end{align*}
\medskip

\noindent \textbf{Term (c)}:
\begin{align*}
    & \sum_w \sum_x m_1(x,z) \, p_{X \mid D=1, Z=z, W=w}(x) \, \big(\mathbf{1}_{\{W=w\}}-p_{W \mid D=1}(w)\big) \\
    &= \sum_x m_1(x,z) \ p_{Z,W \mid D=1}(z,w) - \sum_w \sum_x m_1(x,z) \ p_{X \mid D = 1, Z = z, W = w}(x) \ p_{W \mid D=1}(w) \\
    &= E[m_1(x,z)|D=1,Z=z,W] - \gamma^B_1(z)
\end{align*}

Now putting terms (a), (b), and (c) back into the expression for $\mathbb{IF}(\gamma^B_1(z))$ above:
\begin{align*}
    & \mathbb{IF}(\gamma^B_1(z)) \\ &= \frac{D}{p_D(1)} \Biggl\{ \frac{\mathbf{1}_{\{Z=z\}}}{p_{Z \mid D=1, W}(z \mid w)} (Y_1-Y_0-m_1(X,z)) \\
    & \quad +  \frac{\mathbf{1}_{\{Z=z\}}}{p_{Z \mid D=1, W}(z \mid w)} ( m_1(X,z) - E[m_1(X,z)|D=1,Z=z,W)] \\
    & \quad + E[m_1(x,z)|D=1,Z=z,W] - \gamma^B_1(z) \Biggr\} \\
    &= \frac{D}{p_D(1)} \Biggl\{ \frac{\mathbf{1}_{\{Z=z\}}}{p_{Z \mid D=1, W}(z \mid w)} (Y_1-Y_0-m_1(X,z) + m_1(X,z)) \\
    & \quad +  \left( 1-\frac{\mathbf{1}_{\{Z=z\}}}{p_{Z \mid D=1, W}(z \mid w)} \right) \ E[m_1(X,z)|D=1,Z=z,W]  - \gamma^B_1(z) \Biggr\} \\
    &=  \frac{\mathbf{1}_{\{Z=z\}} \ D}{p_{Z \mid D=1, W}(z \mid w) \ p_D(1)} (Y_1 - Y_0) \\
    & \quad + \left( \frac{D}{p_D(1)}  - \frac{\mathbf{1}_{\{Z=z\}} \ D}{p_{Z \mid D=1, W}(z \mid w) \ p_D(1)}\right)  \ E[m_1(X,z)|D=1,Z=z,W] \\
    & \quad - \frac{D}{p_D(1)} \gamma^B_1(z)
\end{align*}

\bigskip

\subsubsection{Derivation of $\gamma^B_0(z)$}

$\gamma^B_0(z)$ can be written as: 
\begin{align*}
    \gamma^B_0(z) = \int \left[ \int m_0(x, z) \, dP_{X \mid D=1, Z=z, W=w}(x) \right]
\, dP_{W \mid D=1}(w).
\end{align*}
We again start by assuming that the data are discrete to simplify the exposition: 
\begin{align*}
    \gamma^B_0(z) = \sum_w \sum_x m_0(x,z) \ p_{X \mid D = 1, Z = z, W = w}(x) \ p_{W \mid D=1}(w)
\end{align*}
Now we use product rule for derivatives, $\mathbb{IF}(\gamma_1\gamma_2)=\mathbb{IF}(\gamma_1)\gamma_2+\gamma_1\mathbb{IF(\gamma_2)}$:
\begin{align*}
    & \mathbb{IF}(\gamma^B_0(z)) \\
    &= \sum_w \sum_x \mathbb{IF}(m_0(x,z)) \ p_{X \mid D = 1, Z = z, W = w}(x) \ p_{W \mid D=1}(w) \\
    & \quad + \sum_w \sum_x m_0(x,z) \ \mathbb{IF}(p_{X \mid D = 1, Z = z, W = w}(x)) \ p_{W \mid D=1}(w) \\
    & \quad + \sum_w \sum_x m_0(x,z) \ p_{X \mid D = 1, Z = z, W = w}(x) \ \mathbb{IF}(p_{W \mid D=1}(w)) \\
    & = \underbrace{\sum_w \sum_x \frac{\mathbf{1}_{\{D=0,X=x,Z=z\}}}{p_{D,X,Z}(0,x,z)} (Y_1-Y_0-m_0(x,z)) \ p_{X \mid D=1, Z=z, W=w}(x) \ p_{W \mid D=1}(w)}_{\text{Term (d)}} \\
    &\quad + \underbrace{\sum_w \sum_x m_0(x,z) \, \frac{\mathbf{1}_{\{Z=z,W=w\}} \ D}{p_{D,Z,W}(1,z,w)} \big(\mathbf{1}_{\{X=x\}}-p_{X \mid D=1, Z=z, W=w}(x)\big) \, p_{W \mid D=1}(w)}_{\text{Term (e)}} \\
    &\quad + \underbrace{\sum_w \sum_x m_0(x,z) \ p_{X \mid D=1, Z=z, W=w}(x) \ \frac{D}{p_D(1)} \big(\mathbf{1}_{\{W=w\}}-p_{W \mid D=1}(w)\big)}_{\text{Term (f)}}.
\end{align*}
We again analyze each of the terms separately. \medskip

\noindent \textbf{Term (d)}:
\begin{align*}
    & \sum_w \sum_x \frac{\mathbf{1}_{\{D=0,X=x,Z=z\}}}{p_{D,X,Z}(0,x,z)} (Y_1-Y_0-m_0(x,z)) \ p_{X \mid D=1, Z=z, W=w}(x) \ p_{W \mid D=1}(w) \\
    &= \frac{(1-D)}{p_D(0)} \sum_w \frac{\mathbf{1}_{\{Z=z\}}}{p_{X,Z \mid D=0}(x,z)} \ (Y_1 - Y_0 - m_0(X,z)) \ p_{X \mid D = 1, Z = z, W = w}(x) \ p_{W \mid D=1}(w) \\
    &= \frac{(1-D)}{p_D(0)} \frac{\mathbf{1}_{\{Z=z\}}}{p_{X,Z \mid D=0}(x,z)} \ (Y_1 - Y_0 - m_0(X,z)) \sum_w \mathbf{1}_{\{W=w\}} p_{X \mid D = 1, Z = z, W = w}(x) \ p_{W \mid D=1}(w) \\
    &= \frac{(1-D)}{p_D(0)} \frac{\mathbf{1}_{\{Z=z\}}}{p_{X,Z \mid D=0}(x,z)} \ (Y_1 - Y_0 - m_0(X,z)) \ p_{V \mid D=1, Z=z, W}(v \mid w) \ p_{W \mid D=1}(w) \cdot \frac{p_{W \mid D=1, Z=z}(w)}{p_{W \mid D=1, Z=z}(w)} \\
    &= \frac{(1-D)}{p_D(0)} \frac{\mathbf{1}_{\{Z=z\}}}{p_{X,Z \mid D=0}(x,z)} \ (Y_1 - Y_0 - m_0(X,z)) \ p_{X \mid D=1, Z=z}(x) \frac{p_{W \mid D=1}(w)}{p_{W \mid D=1, Z=z}(w)} \\
    &=\frac{\mathbf{1}_{\{Z=z\}} (1-D)}{p_{Z \mid D=0}(z)p_D(0)} (Y_1 - Y_0 - m_0(X,z)) \frac{p_{X \mid D=1, Z=z}(x)}{p_{X \mid D=0, Z=z}(x)} \ \frac{p_{W \mid D=1}(w)}{p_{W \mid D=1, Z=z}(w)} \\
    &= \frac{\mathbf{1}_{\{Z=z\}} (1-D)}{p_{Z \mid D=0}(z)p_D(0)} (Y_1 - Y_0 - m_0(X,z)) \frac{p_{D \mid X, Z=z}(1 \mid x)}{p_{D \mid X, Z=z}(0 \mid x)} \frac{p_{D \mid Z=z}(0)}{p_{D \mid Z=z}(1)} \frac{p_{Z \mid D=1}(z)}{p_{Z \mid D=1, W}(z \mid w)}
\end{align*}
Using for the second equality that $X=(V,W)$ and that
\begin{align*}
    p_{X \mid D=1, Z=z, W=w}(x) = p_{V,W \mid D=1, Z=z, W=w}(v,w)
    = \left\{
    \begin{aligned}
    &p_{V \mid D=1, Z=z, W}(v \mid w)  &&\text{if } w = W,\\[4pt]
    &0                      &&\text{if } w \neq W,
    \end{aligned}
    \right.
\end{align*}
and the last equality follows from
\begin{align*}
    \frac{p_{X \mid D=1, Z=z}(x)}{p_{X \mid D=0, Z=z}(x)} &= \frac{\frac{p_{D \mid X, Z=z}(1 \mid x) \ p_{X \mid Z=z}(x)}{p_{D \mid Z=z}(1)}}{\frac{p_{D \mid X, Z=z}(0 \mid x)p_{X \mid Z=z}(x)}{p_{D \mid Z=z}(0)}} = \frac{p_{D \mid X, Z=z}(1 \mid x) \ p_{D \mid Z=z}(0)}{p_{D \mid X, Z=z}(0 \mid x) \ p_{D \mid Z=z}(1)}.
\end{align*}
Note that we here (as in many other places, check this) need positivity: $p(D=d|Z=z)>0$ and $p(D=d|X,Z=z)>0$.
The last equality follows from 
\begin{align*}
    & p_{W \mid D=1, Z=z}(w)=\frac{p_{W,Z \mid D=1}(w,z)}{p_{Z \mid D=1}(z)}=\frac{p_{Z \mid D=1, W}(z \mid w) \ p_{W \mid D=1}(w)}{p_{Z \mid D=1}(z)} \\
    & \Leftrightarrow \frac{p_{W \mid D=1}(w)}{p_{W \mid D=1, Z=z}(w)}=\frac{p_{Z \mid D=1}(z)}{p_{Z \mid D=1, W}(z \mid w)}
\end{align*}

\noindent \textbf{Term (e)}:
\begin{align*}
    & \sum_w \sum_x m_0(x,z) \, \frac{\mathbf{1}_{\{Z=z,W=w\}} \ D}{p_{D,Z,W}(1,z,w)} \big(\mathbf{1}_{\{X=x\}}-p_{X \mid D=1, Z=z, W=w}(x)\big) \, p_{W \mid D=1}(w) \\
    &= m_0(X,z) \frac{\mathbf{1}_{\{Z=z\}} \ D}{p_{D,Z,W}(1,z,w)} \ p_{W \mid D=1}(w) \\
    & \quad - \sum_x m_0(x,z) \frac{\mathbf{1}_{\{Z=z\}} \ D}{p_{D,Z,W}(1,z,w)} p_{X \mid D=1, Z=z, W}(x \mid w) \ p_{W \mid D=1}(w) \\
    &= m_0(X,z) \frac{\mathbf{1}_{\{Z=z\}} \ D}{p_{D,Z,W}(1,z,w)} \ p_{W \mid D=1}(w) \\
    & \quad - \frac{\mathbf{1}_{\{Z=z\}} \ D}{p_{D,Z,W}(1,z,w)} \ p_{W \mid D=1}(w)  \sum_x m_0(x,z) \  p_{X \mid D=1, Z=z, W}(x \mid w)\\
    &= \frac{\mathbf{1}_{\{Z=z\}} \ D}{p_{D,Z,W}(1,z,w)} \ p_{W \mid D=1}(w) (m_0(X,z)-E[m_0(X,z)|D=1,Z=z,W])
\end{align*}
\medskip

\noindent \textbf{Term (f)}:
\begin{align*}
    & \sum_w \sum_x m_0(x,z) \ p_{X \mid D=1, Z=z, W=w}(x) \ \frac{D}{p_D(1)} \big(\mathbf{1}_{\{W=w\}}-p_{W \mid D=1}(w)\big) \\
    &= \sum_x m_0(x,z) \ p_{X \mid D=1, Z=z, W}(x \mid w) \ \frac{D}{p_D(1)} \\
    & \quad - \sum_w \sum_x m_0(x,z) \ p_{X \mid D=1, Z=z, W=w}(x) \ \frac{D}{p_D(1)} \ p_{W \mid D=1}(w) \\
    &= \frac{D}{p_D(1)} \left( E[m_0(x,z)|D=1,Z=z,W] - \sum_w \sum_x m_0(x,z) p_{X \mid D=1, Z=z, W=w}(x)  \ p_{W \mid D=1}(w) \right) \\
    &= \frac{D}{p_D(1)} (E[m_0(x,z)|D=1,Z=z,W] - \gamma^B_0(z))
\end{align*}

Now putting terms (d), (e), and (f) back into the expression for $\mathbb{IF}(\gamma^B_0(z))$:
\begin{align*}
    & \mathbb{IF}(\gamma^B_0(z)) \\
    &= \frac{\mathbf{1}_{\{Z=z\}} (1-D)}{p_{Z \mid D=0}(z)p_D(0)} (Y_1 - Y_0 - m_0(X,z)) \frac{p_{D \mid X, Z=z}(1 \mid x)}{p_{D \mid X, Z=z}(0 \mid x)} \frac{p_{D \mid Z=z}(0)}{p_{D \mid Z=z}(1)} \frac{p_{Z \mid D=1}(z)}{p_{Z \mid D=1, W}(z \mid w)} \\
    & \quad + \frac{\mathbf{1}_{\{Z=z\}} \ D}{p_{D,Z,W}(1,z,w)} \ p_{W \mid D=1}(w) (m_0(X,z)-E[m_0(X,z)|D=1,Z=z,W]) \\
    & \quad + \frac{D}{p_D(1)} (E[m_0(x,z)|D=1,Z=z,W] - \gamma^B_0(z)) \\
    &= \frac{\mathbf{1}_{\{Z=z\}} (1-D)}{p_{Z \mid D=0}(z)p_D(0)} \ \frac{p_{D \mid X, Z=z}(1 \mid x)}{p_{D \mid X, Z=z}(0 \mid x)} \frac{p_{D \mid Z=z}(0)}{p_{D \mid Z=z}(1)} \frac{p_{Z \mid D=1}(z)}{p_{Z \mid D=1, W}(z \mid w)} \cdot (Y_1 - Y_0) \\
    & \quad + \left( \frac{\mathbf{1}_{\{Z=z\}} \ D}{p_D(1) \ p_{Z \mid D=1, W}(z \mid w)} - \frac{\mathbf{1}_{\{Z=z\}} (1-D)}{p_{Z \mid D=0}(z)p_D(0)} \ \frac{p_{D \mid X, Z=z}(1 \mid x)}{p_{D \mid X, Z=z}(0 \mid x)} \frac{p_{D \mid Z=z}(0)}{p_{D \mid Z=z}(1)} \frac{p_{Z \mid D=1}(z)}{p_{Z \mid D=1, W}(z \mid w)}   \right) \cdot m_0(X,z) \\
    & \quad + \left( \frac{D}{p_D(1)} - \frac{\mathbf{1}_{\{Z=z\}} \ D}{p_D(1) \ p_{Z \mid D=1, W}(z \mid w)} \right) \cdot E[m_0(X,z)|D=1,Z=z,W] \\
    & \quad - \frac{D}{p_D(1)} \gamma^B_0(z)
\end{align*}

\bigskip

\subsubsection{Combined IF}

Now we can combine the results of the previous two subsections due to linearity to obtain $\mathbb{IF}(\gamma^B(z))$:
\begin{align*}
    & \mathbb{IF}(\gamma^B(z)) = \mathbb{IF}(\gamma^B_1(z)) - \mathbb{IF}(\gamma^B_0(z)) \\
    &= \frac{\mathbf{1}_{\{Z=z\}} \ D}{p_{Z \mid D=1, W}(z \mid w) \ p_D(1)} (Y_1 - Y_0) \\
    & \quad + \left( \frac{D}{p_D(1)}  - \frac{\mathbf{1}_{\{Z=z\}} \ D}{p_{Z \mid D=1, W}(z \mid w) \ p_D(1)}\right)  \ E[m_1(X,z)|D=1,Z=z,W] - \frac{D}{p_D(1)} \gamma^B_1(z) \\
    & \quad - \frac{\mathbf{1}_{\{Z=z\}} (1-D)}{p_{Z \mid D=0}(z)p_D(0)} \ \frac{p_{D \mid X, Z=z}(1 \mid x)}{p_{D \mid X, Z=z}(0 \mid x)} \frac{p_{D \mid Z=z}(0)}{p_{D \mid Z=z}(1)} \frac{p_{Z \mid D=1}(z)}{p_{Z \mid D=1, W}(z \mid w)} \cdot (Y_1 - Y_0) \\
    & \quad - \left( \frac{\mathbf{1}_{\{Z=z\}} \ D}{p_D(1) \ p_{Z \mid D=1, W}(z \mid w)} - \frac{\mathbf{1}_{\{Z=z\}} (1-D)}{p_{Z \mid D=0}(z)p_D(0)} \ \frac{p_{D \mid X, Z=z}(1 \mid x)}{p_{D \mid X, Z=z}(0 \mid x)} \frac{p_{D \mid Z=z}(0)}{p_{D \mid Z=z}(1)} \frac{p_{Z \mid D=1}(z)}{p_{Z \mid D=1, W}(z \mid w)}   \right) \cdot m_0(X,z) \\
    & \quad - \left( \frac{D}{p_D(1)} - \frac{\mathbf{1}_{\{Z=z\}} \ D}{p_D(1) \ p_{Z \mid D=1, W}(z \mid w)} \right) \cdot E[m_0(X,z)|D=1,Z=z,W] \\
    & \quad + \frac{D}{p_D(1)} \gamma^B_0(z) \\
    &= \omega_1(z;O) \cdot \left\{ Y_1 - Y_0 - m_0(X,z) \right\} \\
    & \quad + \omega_2(z;O) \cdot \left\{ E[m_1(X,z)|D=1,Z=z,W] - E[m_0(X,z)|D=1,Z=z,W] \right\} \\
    & \quad - \frac{D}{p_D(1)} \cdot \gamma^B(z),
\end{align*}
where
\begin{align*}
    & \omega_1(z;O) = \frac{\mathbf{1}_{\{Z=z\}} \ D}{p_D(1) \ p_{Z \mid D=1, W}(z \mid w)} 
    - \frac{\mathbf{1}_{\{Z=z\}} (1-D)}{p_{Z \mid D=0}(z)\, p_D(0)} 
    \ \frac{p_{D \mid X, Z=z}(1 \mid x)}{p_{D \mid X, Z=z}(0 \mid x)} 
    \frac{p_{D \mid Z=z}(0)}{p_{D \mid Z=z}(1)} 
    \frac{p_{Z \mid D=1}(z)}{p_{Z \mid D=1, W}(z \mid w)}\\
    &= \frac{\mathbf{1}_{\{Z=z\}} D}{p_D(1)\, p_{Z \mid D=1, W}(z \mid w)} 
    - \frac{\mathbf{1}_{\{Z=z\}} (1-D)}{p_D(0)\, p_{Z \mid D=0}(z)} 
    \frac{p_{D \mid X, Z=z}(1 \mid x)}{p_{D \mid X, Z=z}(0 \mid x)} 
    \frac{p_{D \mid Z=z}(0)}{p_{D \mid Z=z}(1)} 
    \frac{p_{Z \mid D=1}(z)}{p_{Z \mid D=1, W}(z \mid w)} \\
    &= \frac{\mathbf{1}_{\{Z=z\}} D}{p_D(1)\, p_{Z \mid D=1, W}(z \mid w)} 
    - \frac{\mathbf{1}_{\{Z=z\}} (1-D)}{p_{Z \mid D=1, W}(z \mid w)} 
    \frac{p_{D \mid X, Z=z}(1 \mid x)}{p_{D \mid X, Z=z}(0 \mid x)} 
    \frac{p_{D \mid Z=z}(0)\, p_{Z \mid D=1}(z)}{p_D(0)\, p_{Z \mid D=0}(z)\, p_{D \mid Z=z}(1)} \\
    &= \frac{\mathbf{1}_{\{Z=z\}} D}{p_D(1)\, p_{Z \mid D=1, W}(z \mid w)} 
    - \frac{\mathbf{1}_{\{Z=z\}} (1-D)}{p_{Z \mid D=1, W}(z \mid w)} 
    \frac{p_{D \mid X, Z=z}(1 \mid x)}{p_{D \mid X, Z=z}(0 \mid x)} 
    \frac{\frac{p_{Z \mid D=0}(z)\, p_D(0)}{p_Z(z)} \cdot p_{Z \mid D=1}(z)}{p_D(0)\, p_{Z \mid D=0}(z)\, \frac{p_{Z \mid D=1}(z)\, p_D(1)}{p_Z(z)}} \\
    &= \frac{\mathbf{1}_{\{Z=z\}} D}{p_D(1)\, p_{Z \mid D=1, W}(z \mid w)} 
    - \frac{\mathbf{1}_{\{Z=z\}} (1-D)}{p_{Z \mid D=1, W}(z \mid w)} 
    \frac{p_{D \mid X, Z=z}(1 \mid x)}{p_{D \mid X, Z=z}(0 \mid x)\, p_D(1)}\\
    &= \frac{\mathbf{1}_{\{Z=z\}}}{p_{Z \mid D=1, W}(z \mid w)\, p_D(1)} 
    \left(
    D - (1-D)\, \frac{p_{D \mid X, Z=z}(1 \mid x)}{p_{D \mid X, Z=z}(0 \mid x)}
    \right),
\end{align*}
and
\begin{align*}
    \omega_2(z;O) = \frac{D}{p_D(1)} - \frac{\mathbf{1}_{\{Z=z\}} \ D}{p_D(1) \ p_{Z \mid D=1, W}(z \mid w)}.
\end{align*}

\newpage

\section{Asymptotic properties}
\label{app:asymptotic_prop}

This section closely follows arguments from \cite{kennedy2024semiparametric}'s Section 4.
Using the von Mises expansion and that $\phi(h;P)$: influence function (first-order derivative in distribution space) with $\int \phi(z;P)dP(z)=0$ and $\int \phi(z;P)^2dP(z) < \infty$, we can write:
\begin{align*}
    \gamma^{B}(\hat P) - \gamma^{B}(P) &= \int \phi(h;P)\, d(\hat P - P)(h) + R_2(\hat P,P) \\
    &= -\int \phi(o;\hat{P})dP(o)+R_2(\hat P,P)
\end{align*}
where $\hat P$ is the empirical distribution, $P$ the true distribution, $\gamma^{B}(P)$: target functional (here, $\text{BGATT}$), and $R_2(\hat P,P)$: second-order remainder.

Second equality indicates bias of the plug-in estimator $\hat{\gamma}^B_{\text{plug-in}}=\gamma^{B}(\hat P)$. 
We obtain the so-called one-step estimator by correcting for this bias: $\hat{\gamma}^B=\gamma^{B}(\hat P) + P_n\{\phi(O;\hat{P})\}$.
This can be rewritten as follows:
\begin{align*}
    \hat{\gamma}^B-\gamma^{B}(\hat P) &= P_n\{\phi(O;\hat{P})\} \\
    &= S^* + T_1 + T_2 
\end{align*}
where $S^*=(P_n-P)\{\phi(H;P)\}$, $T_1=(P_n-P)\{\phi(H;\hat{P}-\phi(Z;P)\}$ and $T_2=R_2(\hat{P},P) = \gamma^B(\hat{P})-\gamma(P) + \int \phi(h;\hat{P})dP(h)$ (see \cite{kennedy2024semiparametric}).

We want $R_1$ and $R_2$ to be of order $o_P(1/\sqrt{n})$ such that the sample average $S^*$ dominates, since then we have
\begin{align*}
    \sqrt{n}(\hat{\gamma}^B-\gamma) = \sqrt{n}S^*+o_P(1) \rightsquigarrow N(0,\text{var}\{\phi(H;P)\}
\end{align*}
meaning that $\hat{\gamma}^B$ is root-n consistent, asymptotically normal and efficient.

The empirical process term $T_1$ can be taken care of by cross-fitting, as long as the number of folds is finite and $\phi(h;\hat{P})$ is consistent ($\| \phi(h;\hat{P}_{-k})-\phi(h;P) \|=o_P(1)$), see for example Proposition 1 of \cite{kennedy2024semiparametric}. 

Now it's left to show that the remainder bias term $T_2=R_2(\hat{P},P)$ is negligible. 

\small
\begin{align*}
    & R_2(\hat{P},P) = \gamma(\hat{P}) - \gamma(P) + E_P\left[\omega_1(z;O) \{ Y_1 - Y_0 - \hat{m_0}(X,z)\}\right] \\
    & \quad + E_P\left[ \omega_2(z;O) E_{\hat{P}}\left[ \hat{m_1} - \hat{m_0} | D=1,Z=z,W\right] \right] - E_P \left[ \frac{D}{P(D=1)} \gamma(\hat{P}) \right] \\
    &= - \gamma(P) + E_P \left[ \frac{\mathbf{1}[Z=z]}{\hat{p}(Z=z|D=1,W)P(D=1)} \left\{ D-(1-D) \frac{\hat{P}(D=1|X,Z=z)}{1-\hat{P}(D=1|X,Z=z)} \{ Y_1 - Y_0 - \hat{m_0}(X,z)\} \right\} \right] \\
    &+ E_P\left[ \left\{ \frac{D}{P(D=1)} - \frac{\mathbf{1}[Z=z]D}{P(D=1) \hat{P}(Z=z|D=1,W)}\right\} E_{\hat{P}}\left[ \hat{m_1}-\hat{m_0}|D=1,Z=z,W\right]\right] \\
    &= - \gamma(P) + E_P \left[ \frac{\mathbf{1}[Z=z]}{\hat{p}(Z=z|D=1,W)P(D=1)} \left\{ D-(1-D) \frac{\hat{P}(D=1|X,Z=z)}{1-\hat{P}(D=1|X,Z=z)} \{ Y_1 - Y_0 - m_0(X,z)\} \right\} \right] \\
    &+ E_P \left[ \frac{\mathbf{1}[Z=z]}{\hat{p}(Z=z|D=1,W)P(D=1)} \left\{ D-(1-D) \frac{\hat{P}(D=1|X,Z=z)}{1-\hat{P}(D=1|X,Z=z)} \{ m_0 - \hat{m_0}\} \right\} \right] \\ 
    &+ E_P\left[ \left\{ \frac{D}{P(D=1)} - \frac{\mathbf{1}[Z=z]D}{P(D=1) \hat{P}(Z=z|D=1,W)}\right\} E_{\hat{P}}\left[ \hat{m_1}-\hat{m_0}|D=1,Z=z,W\right]\right] \\
    &= - \gamma(P) + E_P \left[ \frac{P(Z=z|D=1,W)}{\hat{P}(Z=z|D=1,W)} E_P\left[ m_1-m_0|D=1,W,Z=z \right] | D=1  \right] \\
    &+ E_P \left[ \frac{\mathbf{1}[Z=z]}{\hat{P}(Z=z|D=1,W)P(D=1)} \left\{m_0-\hat{m_0} \right\} \right. \\
    & \cdot \left. \left\{ E_P[D|X,Z=z]-\frac{\hat{P}(D=1|X,Z=z)}{1-\hat{P}(D=1|X,Z=z)} E_P[1-D|X,Z=z] \right\} \right] \\
    &+  E_P\left[ \left\{ \frac{D}{P(D=1)} - \frac{\mathbf{1}[Z=z]D}{P(D=1) \hat{P}(Z=z|D=1,W)}\right\} E_{\hat{P}}\left[ \hat{m_1}-\hat{m_0}|D=1,Z=z,W\right]\right] \\
    &+ E_P\left[ \left\{ \frac{D}{P(D=1)} - \frac{\mathbf{1}[Z=z]D}{P(D=1) \hat{P}(Z=z|D=1,W)}\right\} E_{\hat{P}}\left[ \hat{m_1}-\hat{m_0}|D=1,Z=z,W\right]\right] \\
    &= - \gamma(P) + E_P \left[ \frac{P(Z=z|D=1,W)}{\hat{P}(Z=z|D=1,W)} E_P\left[ m_1-m_0|D=1,W,Z=z \right] | D=1  \right] \\
    &+ E_P \Bigg[ \frac{\mathbf{1}[Z=z]}{\hat{P}(Z=z|D=1,W)P(D=1)} \left\{m_0-\hat{m_0} \right\}  \\
    & \cdot \left\{\frac{P(D=1|X,Z=z)\cdot(1-\hat{P}(D=1|X,Z=z))-\hat{P}(D=1|X,Z=z)(1-P(D=1|X,Z=z)}{1-\hat{P}(D=1|X,Z=z)}\right\}  \Bigg]\\
    &+ E_P\left[ \left\{ \frac{D}{P(D=1)} - \frac{\mathbf{1}[Z=z]D}{P(D=1) \hat{P}(Z=z|D=1,W)}\right\} E_{\hat{P}}\left[ \hat{m_1}-\hat{m_0}|D=1,Z=z,W\right]\right] \\
    &= E_P\left[ E_{\hat{P}} [\hat{m_1} -\hat{m_0} |D=1,Z=z,W]|D=1 \right] \\
    &- E_P \left[ \frac{P(Z=z|D=1,W)}{\hat{P}(Z=z|D=1,W)} E_{\hat{P}}[\hat{m}_1-\hat{m}_0|D=1,Z=z,W] |D=1  \right] \\
    & - E_P \left[ E_P[m_1-m_0|D=1,Z=z,W ] | D=1\right]\\
    &+ E_P \left[ \frac{P(Z=z|D=1,W)}{\hat{P}(Z=z|D=1,W)} E_P[m_1-m_0|D=1,W,Z=z] | D=1 \right] \\
    &+ E_P\left[ \frac{\mathbf{1}[Z=z]}{\hat{P}(Z=z|D=1,W) P(D=1)} \cdot \{ m_0 -\hat{m}_0 \} \frac{P(D=1|X,Z=z) - \hat{P}(D=1|X,Z=z)}{1-\hat{P}(D=1|X,Z=z)} \right] \\
    &= \underbrace{E_P \left[ \left\{ 1 - \frac{P(Z=z|D=1,W)}{\hat{P}(Z=z|D=1,W)} \right\} \cdot \left\{ E_{\hat{P}}[\hat{m}_1-\hat{m}_0|D=1,Z=z,W] - E_P[m_1-m_0|D=1,Z=z,W] \right\} |D=1 \right]}_{R_{2,A}} \\
    &+ \underbrace{E_P\left[ \frac{\mathbf{1}[Z=z]}{\hat{P}(Z=z|D=1,W)P(D=1)} \frac{\hat{P}(D=1|X,Z=z)}{1-\hat{P}(D=1|X,Z=z)} \cdot \{ \hat{m}_0 -m_0 \} \right]}_{R_{2,B}}
\end{align*}
\normalsize
where the second equality follows since $\gamma(\hat{P})-E_P \left[ \frac{D}{P(D=1)} \gamma(\hat{P}) \right]=0$, the third equality follows since $Y_1-Y_0-\hat{m}_0 = (Y_1-Y_0-m_0)+(m_0-\hat{m}_0)$, the fourth equality from Law of Iterated Expectations, the fifth equality follows from $E_P[D|X,Z=z]=P(D=1|X,Z=z)$ and $E[1-D|X,Z=z]=1-P(D=1|X,Z=z)$.

We now bound the remainder terms $R_{2,A}$ and $R_{2,B}$.

\subsection{Bound $R_{2,A}$}

\small
\begin{align*}
    R_{2,A} &= E_P \Bigg[   \frac{\hat{P}(Z=z|D=1,W) -P(Z=z|D=1,W)}{\hat{P}(Z=z|D=1,W)} \\
    & \cdot  \left\{ E_{\hat{P}}[\hat{m}_1-\hat{m}_0|D=1,Z=z,W] - E_P[m_1-m_0|D=1,Z=z,W] \right\} |D=1 \Bigg] \\
\end{align*}
\normalsize
It holds that 
\begin{align*}
    & \left| R_{2,A} \right| \\
    & \leq E_P \Bigg[  \left| \frac{\hat{P}(Z=z|D=1,W) -P(Z=z|D=1,W)}{\hat{P}(Z=z|D=1,W)} \right| \\
    & \cdot \left| \left\{ E_{\hat{P}}[\hat{m}_1-\hat{m}_0|D=1,Z=z,W] - E_P[m_1-m_0|D=1,Z=z,W] \right\} \right| |D=1 \Bigg] 
\end{align*}
Assume that $\hat{P}(Z=z|D=1,W) \geq \varepsilon_A$ such that 
\begin{align*}
    \left| \frac{\hat{P}(Z=z|D=1,W) -P(Z=z|D=1,W)}{\hat{P}(Z=z|D=1,W)} \right| \leq \frac{1}{\varepsilon_A} \left|\hat{P}(Z=z|D=1,W) -P(Z=z|D=1,W)\right|
\end{align*}
Then, by Cauchy-Schwarz:
\begin{align*}
    & \left| R_{2,A} \right| \\
    & \leq \frac{1}{\varepsilon_A} E_P \Bigg[  \left|\hat{P}(Z=z|D=1,W) -P(Z=z|D=1,W)\right| \\
    & \cdot \left|  E_{\hat{P}}[\hat{m}_1-\hat{m}_0|D=1,Z=z,W] - E_P[m_1-m_0|D=1,Z=z,W]  \right| |D=1 \Bigg] \\
    & \leq \frac{1}{\varepsilon_A} \| \hat{P}(Z=z|D=1,W) -P(Z=z|D=1,W) \|_{L^2(P \mid D=1)}\\
    &\cdot \|  E_{\hat{P}}[\hat{m}_1-\hat{m}_0|D=1,Z=z,W] - E_P[m_1-m_0|D=1,Z=z,W)] \|_{L^2(P \mid D=1)}
\end{align*}

So for instance, if $\| \hat{P}(Z=z|D=1,W) -P(Z=z|D=1,W) \|_{L^2(P \mid D=1)}=o_P(n^{-1/4})$ and $ \|  E_{\hat{P}}[\hat{m}_1-\hat{m}_0|D=1,Z=z,W] - E_P[m_1-m_0|D=1,Z=z,W]  \|_{L^2(P \mid D=1)}=o_P(n^{-1/4})$, then $R_{2,A}=o_P(1/\sqrt{n})$, meaning that machine learners can be used for estimation of these nuisance functions.

\subsection{Bound $R_{2,B}$}

\begin{align*}
    R_{2,B} = E_P \left[ \frac{\mathbf{1}[Z=z]}{\hat{P}(Z=z|D=1,W)P(D=1)} \frac{\hat{P}(D=1|X,Z=z)-P(D=1|X,Z=z)}{1-\hat{P}(D=1|X,Z=z)} \cdot \{ \hat{m}_0 -m_0 \}\right]
\end{align*}

Then it holds that
\begin{align*}
    R_{2,B} \leq E_P \left[\frac{\mathbf{1}[Z=z]}{\hat{P}(Z=z|D=1,W)P(D=1)}  \left|  \frac{\hat{P}(D=1|X,Z=z)-P(D=1|X,Z=z)}{1-\hat{P}(D=1|X,Z=z)} \right| \cdot \left| \hat{m}_0 -m_0 \right|  \right]
\end{align*}

Now assume that $1-\hat{P}(D=1|X,Z=z)\geq \varepsilon_B$ and the earlier assumption that $\hat{P}(Z=z|D=1,W) \geq \varepsilon_A$, such that
\begin{align*}
    & \frac{1}{\hat{P}(Z=z|D=1,W)} \left| \frac{\hat{P}(D=1|X,Z=z)-P(D=1|X,Z=z)}{1-\hat{P}(D=1|X,Z=z)} \right| \\
    &\leq \frac{1}{\varepsilon_A \varepsilon_B} \left| \hat{P}(D=1|X,Z=z)-P(D=1|X,Z=z) \right|.
\end{align*}

Then, using Cauchy Schwarz: 
\begin{align*}
    R_{2,B} &\leq \frac{1}{\varepsilon_A \varepsilon_B P(D=1)} E_P \left[ \mathbf{1}[Z=z] \cdot \left| \hat{P}(D=1|X,Z=z)-P(D=1|X,Z=z) \right| \cdot \left| \hat{m}_0 -m_0 \right|  \right] \\
    &= \frac{P(Z=z)}{\varepsilon_A \varepsilon_B P(D=1)} E_P \left[ \left| \hat{P}(D=1|X,Z=z)-P(D=1|X,Z=z) \right| \cdot \left| \hat{m}_0 -m_0 \right| |Z=z \right]\\
    &\leq \frac{P(Z=z)}{\varepsilon_A \varepsilon_B P(D=1)} \cdot
    \| \hat{P}(D=1|X,Z=z)-P(D=1|X,Z=z) \|_{L^2(P \mid Z=z)} \cdot
    \| \hat{m}_0 -m_0 \|_{L^2(P \mid Z=z)}
\end{align*}

So for instance, if $\| \hat{P}(D=1|X,Z=z)-P(D=1|X,Z=z) \|_{L^2(P \mid D=1)}=o_P(n^{-1/4})$ and $ \| \hat{m}_0 -m_0 \|_{L^2(P \mid D=1)}=o_P(n^{-1/4})$, then $R_{2,B}=o_P(1/\sqrt{n})$, meaning that machine learners can be used for estimation of these nuisance functions.

\end{appendices}

\end{document}